\documentclass{amsart}

\usepackage{amsfonts,amsmath,amsxtra,amsthm,amssymb,latexsym}
\usepackage{color}

\newtheorem{theorem}{Theorem}[section]
\newtheorem{corollary}[theorem]{Corollary}
\newtheorem{lemma}[theorem]{Lemma}

\newtheorem{definition}[theorem]{Definition}
\newtheorem{remark}[theorem]{Remark}
\newtheorem{example}[theorem]{Example}
\newtheorem{proposition}[theorem]{Proposition}

\numberwithin{equation}{section}

\DeclareMathOperator{\supp}{supp} 
 
 \DeclareMathOperator{\dom}{dom}
 
\DeclareMathOperator{\ess}{ess}

\DeclareMathOperator{\loc}{loc} \DeclareMathOperator{\comp}{comp}

\newcommand{\be}{\begin{equation}}
\newcommand{\ee}{\end{equation}}

\newcommand{\floor}[1]{\lfloor#1 \rfloor}
\newcommand\R{{\mathbb{R}}}

\newcommand\N{{\mathbb{N}}}

\newcommand\gH{{\mathfrak{H}}}
\newcommand\gt{{\mathfrak{t}}}

\newcommand\gb{{\mathfrak{b}}}
\newcommand\gq{{\mathfrak{q}}}

\newcommand{\gd}{{d}}
\newcommand{\gA}{{\alpha}}
\newcommand{\gB}{{\beta}}

\newcommand\cH{{\mathcal{H}}}

\newcommand\rH{{\rm{H}}}

\newcommand\rD{{\rm{d}}}

\begin{document}

\title[Schr\"odinger operators with $\delta'$-interactions]{Spectral theory of semi-bounded  Schr\"odinger  operators with $\delta'$-interactions}

\author[A.\ Kostenko]{Aleksey Kostenko}
\address{Fakult\"at f\"ur Mathematik\\
Universit\"at Wien\\
Nordbergstr. 15\\
1090 Wien, Austria}

\author[M.\ Malamud]{Mark Malamud}
\address{Institute of Applied Mathematics and Mechanics\\
NAS of Ukraine\\ R. Luxemburg str. 74\\
Donetsk 83114\\ Ukraine}

\keywords{Schr\"odinger operator,
$\delta'$-point interaction, quadratic form, spectral properties}
\subjclass[2010]{
34L40;
47E05;
81Q10;
34L05}

\thanks{{\it The research was funded by the Austrian Science Fund (FWF) under project No.\ M1309--N13}}





\begin{abstract}
We study spectral properties of Hamiltonians $\rH_{X,\gB,q}$ with $\delta'$-point interactions
on a discrete set $X=\{x_k\}_{k=1}^\infty\subset\R_+$. 
Using the form approach, we establish analogs of some classical results on operators $\rH_q=-d^2/dx^2+q$ with locally
integrable potentials $q\in L^1_{\loc}(\R_+)$. In particular, we establish analogues of the Glazman--Povzner--Wienholtz theorem, the Molchanov discreteness criterion, and the Birman theorem
on stability of an essential spectrum. It turns out that in contrast to the case of Hamiltonians with $\delta$-interactions, spectral properties of operators $\rH_{X,\gB,q}$ are closely connected with those of $\rH_{X,q}^N=\oplus_{k}\rH_{q,k}^N$, where $\rH_{q,k}^N$ is the Neumann realization
of $-d^2/dx^2+q$ in $L^2(x_{k-1},x_k)$.
\end{abstract}

\maketitle

\tableofcontents

\section{Introduction}\label{intro}

The main object of the paper is the Hamiltonian $\rH_{X,\gB,q}$ associated in $L^2(\R_+)$ with the formal differential expression
\be\label{eq:i01}
\tau_{X,\gB,q}:=-\frac{d^2}{dx^2}+q(x)+\sum_{k=1}^\infty\gB_k(.,\delta'_k)\delta'_k,
\ee
where $\delta'_k:=\delta'(x-x_k)$ is the derivative of a Dirac delta function centered at $x_k$.  It is also assumed that $q\in L^1_{\loc}(\R_+)$ and $X=\{x_k\}_{k=1}^\infty$ is a strictly increasing sequence such that $x_k\uparrow +\infty$.

Hamiltonians with $\delta'$-interactions are known as exactly solvable models of quantum mechanics (see \cite{Alb_Ges_88}). These models called "solvable" since their resolvents can be computed explicitly in terms of the interaction strengths and the location of the sources. As a consequence the spectrum, the eigenfunctions, and further spectral properties can be determined explicitly. Models of this type have been extensively discussed in the physical literature, mainly in atomic, nuclear and solid state physics. Note also that a connection of  Hamiltonians with $\delta'$-interactions and the  so-called Krein--Stieltjes strings  has recently been discovered in \cite{KosMal09}, \cite{KosMal10}.

The existence of the model \eqref{eq:i01} was pointed out in 1980 by Grossmann, Hoegh--Krohn and Mebkhout \cite{GroHoeMeb80}.
However, the first rigorous treatment of \eqref{eq:i01} was made by Gesztesy and Holden in \cite{GesHol87}. Namely, they defined the Hamiltonian $\rH_{X,\gB,q}$ by using the method of boundary conditions. To be precise, the operator $\rH_{X,\gB,q}$ is defined in $L^2(\R_+)$ as the closure of the symmetric operator $\rH_{X,\gB,q}^0$,
\be\label{eq:h_b3}
\rH_{X,\gB,q}:=\overline{\rH_{X,\gB,q}^0}.
\ee
where
 \begin{align}
\rH_{X,\gB,q}^0f:=&\tau_qf=-f''+q(x)f,\quad f\in\dom(\rH_{X,\gB,q}^0),\label{eq:h_b1}\\
\dom(\rH_{X,\gB,q}^0):=\Big\{f\in W^{2,1}_{\comp}&([0,b)\setminus X):\, f'(0)=0,\Big. \nonumber
\\ &\Big.  \begin{array}{c}
f'(x_k+)=f'(x_k-)\\ f(x_k+)-f(x_k-)=\gB_k f'(x_k) \end{array},\, \tau_q f\in L^2(\R_+)\Big\}.\label{eq:h_b2}
\end{align}
Note that in the case $\gB_k=0$, $k\in\N$, the Hamiltonian $\rH_{X,\gB,q}$ coincides with the free Hamiltonian $\rH_q$. If $\gB_k= \infty$, then the boundary condition at $x_k$  reads as
$f'(x_k+)=f'(x_k-) = 0$. Therefore, the operator $\rH_{X,\infty,q}$ becomes
\be\label{eq:h_N}
 \rH_{X,\infty,q}:=\rH_{X,q}^N = \bigoplus_{k\in \N} \rH_{q,k}^N, \quad  \dom(\rH_{X,q}^N) = \bigoplus_{k\in \N}\dom(\rH_{q,k}^N),
\ee
where $\rH_{q,k}^N$ is the Neumann realization of $\tau_q=-\frac{d^2}{dx^2}+q$ in $L^2(x_{k-1},x_k)$. If $q=\bold{0}$, then we set $\rH_{X}^N:=\rH_{X,\bold{0}}^N$ and $\rH_{k}^N:=\rH_{\bold{0},k}^N$.

The literature on point interactions is vast and for a comprehensive review we refer to \cite{Alb_Ges_88}, \cite{Exn05} and \cite{KosMal12}. The main novelty of the present paper is that we use the form approach for the study of spectral properties of $\rH_{X,\gB,q}$.
Note that the form approach has successfully been applied to Hamiltonians with $\delta$-interactions
\be\label{eq:hxa}
\rH_{X,\gA,q}=-\frac{d^2}{dx^2}+q(x)+\sum_{k=1}^\infty\gA_n\delta(x-x_k),\quad \gA_k\in\R,
\ee
(see for instance \cite{AKM_10}, \cite{Bra85} and the references therein).
As distinguished from this case,
up to now it was not clear how to apply the form approach for the study of Hamiltonians with $\delta'$-interactions (cf. \cite[Section 7.2]{Exn08}). Indeed, a very naive guess is to consider  a single $\delta'$-interaction at $x_0$  as the following form sum
\[
\gt'[f]=\int_{\R} |f'(x)|^2\, dx +\gB_0|f'(x_0)|^2
\]
defined on the domain
\[
 \dom(\gt')=\{f\in W^{1,2}(\R):\, f'(x_0) \ \text{exists and is finite}\}.
\]
Clearly, the form $\gt'$ is not closable. However\footnote{In the paper \cite{BehLanLot12}, which appeared during the preparation of the present paper, Hamiltonians with a $\delta'$-interaction supported on a hypersurface are treated in the framework  of the form approach.},
 one needs to consider a $\delta'$-interaction as a form sum of two forms $\gt_N$ and $\gb$, where
\be\label{eq:3.5}
\gt_N[f]:=\int_{\R} |f'(x)|^2\, dx,\quad \dom(\gt_N):=W^{1,2}(\R\setminus\{x_0\}),
\ee
and
\be\label{eq:3.6}
\gb[f]:=\frac{|f(x_0+)-f(x_0-)|^2}{\gB_0},\quad \dom(\gt_N):=W^{1,2}(\R\setminus\{x_0\}).
\ee
Let us note that the operator
\be
\rH_{x_0}^N:=-\frac{d^2}{dx^2}, \quad \dom(\rH_{x_0}^N)=\{f\in W^{2,2}(\R\setminus \{x_0\}):\ f'(x_0+)=f'(x_0-)=0\},
\ee
is associated with the form $\gt_N$. Clearly, $\rH_{x_0}^N$ is the direct sum of Neumann realizations of $-\frac{d^2}{dx^2}$ in $L^2(-\infty,x_0)$ and $L^2(x_0,+\infty)$, respectively. Note that the form $\gb$ is infinitesimally form bounded with respect to the form $\gt_N$ and hence, by the KLMN theorem,  the form
\be\label{eq:3.7}
\gt'[f]:=\gt_N[f]+\gb[f], \quad \dom(\gt'):=W^{1,2}(\R\setminus\{x_0\}),
\ee
is closed and lower semibounded and gives rise to a self-adjoint operator 
\be\label{eq:3.8}
\rH'=-\frac{d^2}{dx^2},\quad \dom(\rH'):=\Big\{f\in W^{2,2}(\R\setminus\{x_0\}): \begin{array}{c}
f'(x_0+)=f'(x_0-)\\
f(x_0+)-f(x_0-)=\gB_0f'(x_0+)\end{array}\Big\}.
\ee
A precise definition of the form in the case of an infinite sequence $X$ is given in Section \ref{ss:2.2}.

 Let us emphasize that the definition of a $\delta'$-interaction via the form sum \eqref{eq:3.7}
 allows to observe the key difference between $\delta$ and $\delta'$ point interactions.
 Namely, the Hamiltonian $\rH_{X,\gA,q}$ with $\delta$-interactions \eqref{eq:hxa} is usually treated as a form perturbation of the free Hamiltonian $\rH_{q}=-d^2/dx^2+q(x)$ in $L^2(\R_+)$ (see \cite{AKM_10}, \cite{Bra85}).
 However, the Hamiltonian $\rH_{X,\gB,q}$ with $\delta'$-interactions on $X$ can be treated  as a form perturbation of the operator $\rH_{X,q}^N$ defined by \eqref{eq:h_N}.
For instance, in the case of infinitely many interaction centers and $q=\bold{0}$, the free Hamiltonian $\rH_{0}$ has
 purely absolutely continuous spectrum although the spectrum of $\rH_{X}^N$ is purely point.
 Let us also mention that the idea to consider Hamiltonians with $\delta'$-interactions
 $\rH_{X,\gB,q}$ as a perturbation (in a resolvent sense) of the Neumann realization
 $\rH_{X,q}^N$ was already used by P. Exner in \cite{Exn95}  to study the  spectral  properties
 of $\delta'$ Wannier--Stark Hamiltonians.


Using the form approach, we establish a number of results on semibounded Hamiltonians $\rH_{X,\gB,q}$.
Let us describe the content of the paper and the main results.

In Section \ref{sec:sb}, we investigate self-adjointness and semiboundedness of Hamiltonians with $\delta'$-interactions.
Our first main result is the analog of the classical Glazman--Povzner--Wienholtz
theorem (see, e.g.,  \cite{Ber68}, \cite{Gla65}, \cite{Wie58}).

%
%
\begin{theorem}\label{th_III.1}
If the minimal operator $\rH_{\min} = \rH_{X,\gB,q}$ is lower
semibounded, then it is self-adjoint, $\rH_{X,\gB,q} =
(\rH_{X,\gB,q})^*$.
\end{theorem}
A similar result for the Hamiltonians $\rH_{X,\gA,q}$ with $\delta$-interactions has been obtained in  \cite{AKM_10} (see also  \cite{HryMyk12} for another proof).

Theorem  \ref{th_III.1}  immediately implies  the following  important statement (see Corollary \ref{cor:2.4}): \emph{if
the form $\gt_{X,\gB,q}$ is lower semibounded, then it is
closable and the self-adjoint operator associated with its closure
coincides with $\rH_{X,\gB,q}$}.

Let us also present some explicit conditions
for the lower semi-boundedness of $\rH_{X,\gB,q}$.
For instance (see Proposition \ref{cor:3.3}), {\em assuming  $\gd^*=\sup_k\gd_k<\infty$,   we show   that the operator $\rH_{X,\gB,q}$ is  lower semibounded whenever  $q$ and
$\gB$ satisfy the following conditions: }
   \begin{equation}\label{I_brinck} 
C_0:=\sup_{k\in\N} \frac{1}{\gd_k}\int_{\Delta_k}q_-(x) dx <
\infty, \qquad\ \
C_1:=\sup_{k\in\N}
\frac{(\gB_k^{-1})^-}{\min\{\gd_k,\gd_{k+1}\}}<\infty.
  \end{equation}
Here
\[
q_-(x):=(|q(x)|- q(x))/2,\quad \text{and} \quad (\gB_k^{-1})^- := \Big(\frac{1}{|\gB_k|}
- \frac{1}{\gB_k}\Big)/2.
\]
 \emph{Emphasize that \eqref{I_brinck} is also necessary
if both $q$ and $\gB$ are negative}.

Combining this statement with Theorem \ref{th_III.1} yields the self-adjointness  of $\rH_{X,\gB,q}$ under the conditions \eqref{I_brinck}. In particular,
$\rH_{X,\gB,q}$ is self-adjoint and lower semibounded if $q$ is lower semibounded and $\gB$ is
nonnegative.

Section \ref{sec:III} is devoted to the problem of discreteness of $\sigma(\rH_{X,\gB,q})$.
The main result is the following analog of the classical  Molchanov discreteness criterion
\cite{Mol53}, \cite{Bri59}, \cite{Gla65} (see also \cite{AKM_10} for  the case of Hamiltonians \eqref{eq:hxa}).
      \begin{theorem}\label{th_discretcriter}
Let  
$q\in L^1_{\loc}(\R_+)$,   $d^*<\infty$  
 and conditions  \eqref{I_brinck} be  satisfied.  

\begin{itemize}
\item[(i)] If the spectrum $\sigma(\rH_{X, \beta,q})$ of the (lower semibounded) Hamiltonian $\rH_{X,\beta,q}$  
is discrete, then the following two conditions are satisfied:
    \begin{equation}\label{Intro2.1}
\lim_{x\to\infty}\int^{x + \varepsilon}_x q(t)dt   
= \infty \qquad \text{for every}\ \ \varepsilon>0,
    \end{equation}
and
%
       \begin{equation}\label{Intro2.2}
\lim_{k\to\infty}\frac{1}{\gd_k} \left(\int_{\Delta_k}q(x) dx +  \Big(\frac{1}{\gB_{k-1}} +
\frac{1}{\gB_k}\Big) \right) = \infty. 
    \end{equation}
    \item[(ii)] If $q$ satisfies condition \eqref{Intro2.1} and
    \begin{equation}\label{Intro2.2B}
\lim_{k\to\infty}\frac{1}{\gd_k} \int_{\Delta_k}q(x) dx  = \infty, 
    \end{equation}
    then the spectrum of $\rH_{X,\gB,q}$ is discrete.
    \end{itemize}
    \end{theorem}
Let us mention that  {\emph{conditions \eqref{Intro2.1} and \eqref{Intro2.2B} provide a discreteness criterion for
the Neumann realization $\rH_{X,q}^N$ given by \eqref{eq:h_N}}} (see Theorem \ref{th:discr_N}).

Note that conditions \eqref{Intro2.1} and \eqref{Intro2.2} remain necessary under weaker assumptions on $q$ and $\gB$ (see Propositions \ref{prop4.5} and \ref{prop:4.3}). In particular, by Proposition \ref{prop:4.3}, if the Hamiltonian  $\rH_{X,\beta,q}$ is  semibounded from below and condition \eqref{Intro2.2} is violated, then the spectrum of $\rH_{X,\gB,q}$ is not discrete.
Let us also mention that {\em Molchanov's condition \eqref{Intro2.1} remains necessary for the discreteness
of the Hamiltonian  $\rH_{X,\beta,q}$, however, it is no longer sufficient without additional assumptions on $q$ and $\gB$}.
In particular, if
$0<\gd_*\le \gd^*<+\infty$ and $\inf_k(\gB_k^{-1})^-<\infty$ and $q$ satisfies \eqref{I_brinck}, then Molchanov's condition \eqref{Intro2.1} provides a criterion for
the operator $\rH_{X,\gB,q}$ to have discrete spectrum
(see Corollary \ref{cor3.9}). Moreover, it implies that  {\em  the spectrum of $\rH_{X,\gB}$ is not discrete if the Hamiltonian $\rH_{X,\gB}$ is lower semibounded} (Corollary \ref{cor3.8}).

Note that there is a gap  between necessary and sufficient conditions for the (semibounded)
Hamiltonians $\rH_{X,\gB,q}$ to be discrete.
In fact, there are counterexamples showing that the Hamiltonian $\rH_{X,\gB,q}$ has discrete spectrum 
although condition  \eqref{Intro2.2B} does not hold.
Let us also mention that in the case $q\in L^\infty$ the discreteness conditions for non-semibounded Hamiltonians $\rH_{X,\beta,q}$ have been obtained in \cite{KosMal09}, \cite{KosMal10}  using the above mentioned connection with Krein--Stielties strings \cite{KK71}.


In Section \ref{sec:cs}, we show that essential spectra of operators $\rH_{X,\gB,q}$ and $\rH_{X,q}^N$ are closely connected.  Recall that  for a self-adjoint operator $T$ in a Hilbert space $\cH$ {\em the essential spectrum} $\sigma_{\ess}(T)$ of $T$ is the set
\be\label{eq:ess}
\sigma_{\ess}(T):=\{\lambda\in\R: \ {\rm rank}(E_T(\lambda-\varepsilon,\lambda+\varepsilon))=\infty\ \text{for all}\ \varepsilon>0\},
\ee
where $E_T(\cdot)$ is the spectral measure of $T$.

The central result of Section \ref{sec:cs} is the following theorem.

     \begin{theorem}\label{thContSpec}
Assume that $\gd^*<\infty$ and $q\in L^1_{\loc}(\R_+)$ satisfies the first
condition in \eqref{I_brinck}. If
     \begin{equation}\label{eq:0.15}
 \frac{|\beta_k|^{-1}}{\min\{d_k, d_{k+1}\}} \to 0 \qquad
\text{as}\qquad k\to \infty.
       \end{equation}
then
\be
\sigma_{\ess}(\rH_{X,\beta,q})=\sigma_{\ess}(\rH_{X,q}^N).
\ee

If, in addition,
     \begin{equation}
\lim_{k\to \infty}\frac{1}{d_k}\int_{\Delta_k}|q(x)|dx = 0,
     \end{equation}
then
      \begin{equation}\label{1.12}
\sigma_{\ess}(\rH_{X,\beta,q})=
\sigma_{\ess}(\rH_{X}^N).
     \end{equation}
     \end{theorem}
%
%
First of all, let us mention that in contrast to the case of Hamiltonians $\rH_{X,\gA,q}$ with $\delta$-interactions (cf. \cite[Theorem 1.3]{AKM_10}), the spectrum of Hamiltonians $\rH_{X,\gB,q}$ with $\delta'$-interactions remains purely singular under "small" perturbations.

Further, the spectrum of $\rH_{X}^N$ can be explicitly computed in terms of distances $\{\gd_k\}$ and hence we immediately obtain various examples of operators with exotic essential spectra (see Corollaries \ref{cor:5.8}, \ref{cor5.1} and \ref{cor:5.10}). For instance (Corollary \ref{cor:5.10}),
{\em if conditions of Theorem \ref{thContSpec} are satisfied and, in addition, $\lim_{k\to\infty}d_k = 0$,  then
     \begin{equation}\label{1.13}
\sigma_{\ess}(\rH_{X,\beta,q})=\sigma_{\ess}(\rH_{X,q}^N) =
\sigma_{\ess}(\rH_{X}^N)=\{0\},
   \end{equation}
i.e., the spectrum of   $\rH_{X,\beta,q}$  is purely  point and accumulates only at $0$ and $\infty$}. Examples of Schr\"odinger operators with locally integrable potentials and with $\delta$-interactions having exotic essential spectra can be found in \cite{GorHolMol07}, \cite{Ism85}, and \cite{IsmKos10}.

Note that, by Theorem \ref{thContSpec}, in the case $q\ge 0$, the
negative spectrum of the Hamiltonian $\rH_{X,\gB,q}$ is bounded
from below and discrete provided that \eqref{eq:0.15} holds.

Moreover, in Proposition \ref{prop6.3} we generalize one result of Birman on
the $h$-stability (\cite{Bir61}, and \cite[Theorem 2.4]{Gla65}). It
complements Theorem \ref{thContSpec} and shows that, under additional mild
assumptions on $\gB$ and $q$, condition
\eqref{eq:0.15} becomes in a sense  necessary if we replace the
equality $\sigma_{\ess}(\rH_{X, \gB,q})=\sigma_{\ess}(\rH_{X,q}^N)$ by the
family of equalities $\sigma_{\ess}(\rH_{X, h\gB,
q})=\sigma_{\ess}(\rH_{X,q}^N)$, $h\in(0,+\infty)$.

In the final Appendix we collect necessary notions and facts on quadratic forms.

\textbf{Notation:} 
%
%
%
%
Let $ \gd=\{\gd_n\}_{n=1}^\infty\subset\R_+$ be a positive sequence and $p\in[1,\infty]$. By $\ell^p(d)=\ell^p(\N; d)$
we denote the weighted Banach space of sequences $f=\{f_n\}_{n=1}^\infty$
equipped with the norm $\|f\|_{\ell^p(d)}= (\sum_{n=1}^\infty \gd_n|f_n|^p)^{1/p}$.
If $\gd_k=1$ for all $k$, then $\ell^p=\ell^p(\N): = \ell^p(\N; {\bf 1})$.
Let $X$ be a discrete subset of $\R_+$. We shall use the following Sobolev spaces
($n\in\N$):
%
%
\begin{eqnarray*}
&W^{n,p}(\R_+\setminus X):=\{f\in L^p(\R_+): f\in W^{n,p}[x_{k-1},x_k],\  k\in \N,\,  f^{(n)}\in
  L^p(\R_+)\},\\
&W^{n,p}_0(\R_+\setminus X):=\{f\in
W^{n,p}(\R_+): f(x_k)=...=f^{(n-1)}(x_k)=0,\, k\in \N\},\\
&W^{n,p}_{\comp}(\R_+\setminus X):=W^{n,p}(\R_+\setminus X)\cap L^p_{\comp}(\R_+).
\end{eqnarray*}


\section{Semiboundedness and self-adjointness.}\label{sec:sb}

\subsection{The Hamiltonian $\rH_{X,\gB,q}$}\label{ss:2.1}
Let $q=\overline{q}\in L^1_{\loc}(\R_+)$. Assume also that $\gB=\{\gB_k\}_1^\infty\subset \R$ and $X=\{x_k\}_1^\infty\subset \R_+$ is such that $x_k\uparrow +\infty$.  We
set $x_0=0$ and
\begin{align}\label{eq:2.1}
\Delta_k := [x_{k-1}, x_{k}],\quad &d_k := |\Delta_k|=x_k - x_{k-1},\quad  k\in \N,\\
\gd_*:=&\inf_{k\in\N}\gd_k,\qquad \gd^*:=\sup_{k\in\N}\gd_k. \label{eq:2.1B}
\end{align}

 The main object of the paper is the operator $\rH_{X,\gB,q}$ defined in $L^2(\R_+)$ as the closure of the following symmetric operator
\begin{align}\label{eq:h_0}
&\rH^0_{X,\gB,q}f:=\tau_q f=-f''+q(x)f,\quad f\in \dom(\rH_{X,\gB,q}),\\
\dom(\rH_{X,\gB,q}^0)=&\Big\{f\in W^{2,1}_{\comp}(\R_+\setminus X): f'(0)=0, \begin{array}{c}  f'(x_n+)=f'(x_n-)\\
 f(x_n+)-f(x_n-)=\gB_n f'(x_n)\end{array}\Big\}\label{eq:dom_0}.
\end{align}

If $\gB_k=\infty$ for all $k\in \N$, then the conditions at $x_k$ read as
$f'(x_k-) = f'(x_k+)=0$, $k\in \N$, and the operator $\rH_{X,\infty,q}$ takes the form
\be\label{eq:H^N}
\rH_{X,\infty,q}= \rH_{X,q}^N:=\bigoplus_{k\in \N} \rH_{q,k}^N, \quad  \dom(\rH_{X,q}^N) = \bigoplus_{k\in \N}\dom(\rH_{q,k}^N),
\ee
where
$\rH_{q,k}^N$ is the Neumann realization of $\tau_q$ in $L^2[x_{k-1},x_k]$,
\begin{align}\label{eq:HkN1}
&\rH_{q,k}^Nf:=\tau_qf=-f''+q(x)f,\quad f\in\dom(\rH_{q,k}^N),\\
\dom(\rH_{q,k}^N)&=\{f\in W^{2,1}(\Delta_k):\, f'(x_{k-1})=f'(x_k)=0,\ \tau_qf\in L^2(\Delta_k) \}.\label{eq:HkN2}
\end{align}


Next assuming that $q\in L^\infty(\R_+)$, we can specify the description of the domain $\dom(\rH_{X,\gB,q})$
equipped with the graph norm of the operator $\rH_{X,\gB,q}$.  
   \begin{proposition}\label{prop:2.1}
Let $q\in L^\infty(\R_+)$. Then:
\begin{itemize}
\item[(i)]  The operator $\rH_{X,\gB,q}$ is self-adjoint and its domain is given by
 \be
 \dom\big(\rH_{X,\gB,q}\big):=\Big\{f\in  W^{2,2}(\R_+\setminus X):\,  f'(0)=0, 
\begin{array}{c}
f'(x_k+) = f'(x_k-)\\ f(x_k+)-f(x_k-) = \gB_k f(x_k) \end{array}\Big\}.\label{eq:dom2}
 \ee
\item[(ii)] The embedding
\be\label{eq:2.5A}
 W^{2,2}(\R_+\setminus X) \hookrightarrow  W^{1,2}(\R_+\setminus X), 
\ee
holds and is continuous  if and only if $\gd_* >0$.
\item[(iii)] If $ \gd_* >0$, then
the embedding %
  \begin{align}\label{eq:2.5B}
\dom(\rH_{X,\gB,q}) 
\hookrightarrow  W^{1,2}(\R_+\setminus X) 
    \end{align}
holds  and  is continuous.
\end{itemize}
\end{proposition}

     \begin{proof}
    (i) Self-adjointness of $\rH_{X,\gB,q}$ was established in \cite{BusStoWei95} (see also \cite[\S 6]{KosMal09}). Therefore, $\overline{\rH^0_{X,\gB,q}}=(\rH_{X,\gB,q}^0)^*=\rH_{X,\gB,q}$. Noting that $\dom\big( (\rH_{X,\gB,q}^0)^*\big)\subseteq W^{2,2}(\R_+\setminus X)$ and integrating by parts the expression $(\rH_{X,\gB,q}^0f,g)_{L^2}$,  the second claim follows.

(ii)  Let us prove necessity. 
For any sequence of positive numbers $a=\{a_k\}_1^\infty$
set
\be\label{eq:fa}
f_a(x)=\sum_{k\in\N} \frac{a_k}{\gd_k}(x-x_{k-1})\chi_{[x_{k-1},x_k)}(x),\quad x\in\R_+.
\ee
Clearly,  $f''(x)=0$ for all $x\in \R_+\setminus X$ and
\be\label{eq:fa2}
\int_0^\infty|f(x)|^2dx=\frac{1}{2}\sum_{k}a_k^2\gd_k,\qquad \int_0^\infty|f'(x)|^2dx=\sum_{k}\frac{a_k^2}{d_k}.
\ee
Hence $f\in W^{2,2}(\R_+\setminus X)$ if the first series  converges, i.e., $a\in \ell^2(\N; d)$. The second series converges precisely if $a\in \ell^2(\N; d^{-1})$. Therefore, if the embedding \eqref{eq:2.5A} holds, then $\ell^2(\N; d) \hookrightarrow \ell^2(\N; d^{-1})$. Or equivalently, the multiplication operator $M: \ell^2(\N)\to \ell^2(\N)$ given by $(Ma)_k=a_k/d_k$, $k\in\N$, is bounded. The latter implies that $\sup_{k}\gd_k^{-1}=\inf_{k}\gd_k=\gd_*>0$.



Let us show that the condition $\gd_*>0$ is sufficient.
  Firstly, \cite[inequality (IV.1.12)]{Kato66}, implies that for any $n>1$ and $f\in W^{2,2}(\Delta_k)$
\be\label{eq:kato}
\|f'\|_{L^2(\Delta_k)}\le \frac{\gd_k}{n-1}\|f''\|_{L^2(\Delta_k)}+\frac{2n(n+1)}{n-1}\frac{1}{\gd_k}\|f\|_{L^2(\Delta_k)}.
\ee
Further, for all $k\in\N$ there is $n_k\in\N$ such that
\[
1<\gd_*+1\le \gd_k+1\le n_k\le \gd_k+2.
\]
 Therefore,
\[
\frac{\gd_k}{n_k-1}\le 1,\qquad \frac{2n_k(n_k+1)}{n_k-1}\frac{1}{\gd_k}\le 2\frac{n_k+1}{n_k-1}\frac{n_k}{\gd_k}\le 2\Big(1+\frac{2}{\gd_*}\Big)^2=:C_*.
\]
Hence \eqref{eq:kato} with $n=n_k$ implies that
\[
\|f'\|_{L^2(\Delta_k)}\le \|f''\|_{L^2(\Delta_k)}+C_*\|f\|_{L^2(\Delta_k)},\quad k\in \N,
\]
which proves the embedding \eqref{eq:2.5A} if $\gd_*>0$.

(iii) follows by combining \eqref{eq:dom2} with (ii).
   \end{proof}

\begin{remark}
In the case $\gd_*=0$ the embedding \eqref{eq:2.5B} depends on $\gB$ and it might be false. Namely, set $\gB_1=-1/2$, $\gd_k=1/\sqrt{k}$ and $\gB_{k+1}=-\gd_{k}=-1/\sqrt{k}$, $k\in\N$. Further, consider the function $f_a$ defined by \eqref{eq:fa} with $a_k:=\gd_k=1/\sqrt{k}$, $k\in\N$ and set
\[
f(x)=\frac{x^2}{2}\chi_{[0,1]}(x)+f_a(x)\chi_{(1,+\infty)}(x).
\]

It is straightforward to check that $f'(0)=0$, $f'(x_1-)=f'(x_1+)=1$, $f'(x_1+)-f'(x_1-)=-1/2=\gB_1f'(x_1)$, and 
\[
f'(x_k\pm)=1,\quad f(x_k+)-f(x_k-)=0-a_{k-1}=-1/\sqrt{k-1}=\gB_{k}f'(x_k),\ \ k
\ge 2.
\]
Moreover, the first integral in \eqref{eq:fa2} is convergent and hence $f\in L^2(\R_+)$.
However, the second integral in \eqref{eq:fa2} diverges and hence $f'\notin L^2(\R_+)$! Therefore, the embedding \eqref{eq:2.5B} does not hold. 
\end{remark}
 Next we present the proof of Theorem \ref{th_III.1}, which extends the classical Glazman--Povzner--Wienholtz theorem to the case of Hamiltonians with $\delta'$-interactions.
\begin{proof}[Proof of Theorem \ref{th_III.1}]
Without loss of generality we can assume that $\rH_{X,\gB,q}\ge I$.
It suffices to show that $\ker(\rH^*_{X,\gB,q})=\{0\}$, that is the
equation
    \begin{equation}\label{3.14}
\tau_q[u]= -u''(x) + q(x)u(x) = 0,\quad  x\in \R_+\setminus X, \qquad u\in\dom(\rH_{X,\gB,q}^*),
   \end{equation}
has only a trivial solution (derivative is understood in a distributional sense).

Assume the converse, i.e., $u(\cdot)\not\equiv 0$ is such a solution.
Note that we can always construct the function $\varphi_k\in C^\infty_{\comp}(\R_+)$ such that
\begin{equation}\label{eq:phi1}
\supp \varphi_k=[0,x_k+1],\quad 0\le \varphi_k\le 1,\quad \varphi_k(x)=
\begin{cases}
1,& x\in[0,x_k]\\
0,& x\ge x_k+1
\end{cases}
\end{equation}
and
\begin{equation}\label{eq:phi2}
-2\le \varphi_k'\le 0,\quad \varphi_k'(x_j)=0\quad \text{for all}\quad x_j\in[x_k,x_k+1].
\end{equation}
Then we set
\begin{equation}\label{eq:u_k}
u_k(x)=u(x)\varphi_k(x),\quad k\in\N.
\end{equation}

Clearly, $\supp u_k\subset[0,x_k+1]$. It is straightforward to check that $u_k\in\dom(\rH^*_{X,\gB,q})$ and hence
$u_k\in\dom(\rH^0_{X,\gB,q})$.

Further,
    \begin{eqnarray}\label{3.16}
(\rH^0_{X,\gB,q}u_k,u_k)&=&\int^{\infty}_{0} [-u''_k(x) + q(x)u_k(x)]u_k(x)dx \nonumber \\
&=& - \int^{\infty}_{0}[2 u'(x)\varphi'_k(x) + u(x)
\varphi''_k(x)]u(x)\varphi_k(x) dx \nonumber\\
&=& - \int^{x_k+1}_{x_k}[2 u'(x)\varphi'_k(x) + u(x)
\varphi''_k(x)]u(x)\varphi_k(x) dx \\
&\ge& (u_k,u_k)\ge
\int^{x_k}_{0}u^2(x)dx. \nonumber
     \end{eqnarray}
On the other hand, integrating by parts and noting that $\varphi_k$ has a compact support and $\varphi_k'(x_i)=0$ if $x_i\in[x_k,x_k+1]$, we get
    \begin{align}
 &\int^{\infty}_{0} 2u(x)u'(x)\varphi_k(x)\varphi_k'(x) dx =  \int^{x_k+1}_{x_k} 2u(x)u'(x)\varphi_k(x)\varphi_k'(x) dx\nonumber \\
&= \frac{1}{2}\int^{x_k+1}_{x_k}\bigl(u^2(x)\bigr)'\bigl(\varphi^2_k(x)\bigr)'dx
=-\int^{x_k+1}_{x_k}u^2(x)[\varphi''_k(x)\varphi_k(x) + \bigl(\varphi'_k(x)\bigr)^2]dx. \label{3.17}
    \end{align}

Combining \eqref{3.16} with \eqref{3.17}, we obtain
\be\label{3.16A}
(\rH^0_{X,\gB,q}u_n,u_n)=\int^{x_k+1}_{x_k}u^2(x) \bigl(\varphi'_k(x)\bigr)^2dx\le 4\int^{x_k+1}_{x_k}u^2(x) dx.
     \ee
Therefore, we get
     \be\label{3.18}
\int^{x_k}_0 u^2(x)dx \le 4\int^{x_k+1}_{x_k}u^2(x) dx.
     \ee
Noting that $u\in L^2(\R_+)$ and $x_k\to +\infty$, inequality \eqref{3.18} yields $u\equiv 0$. This contradiction completes the proof.
\end{proof}
\begin{remark}
Theorem \ref{th_III.1} generalizes  the celebrated Glazman--Povzner--Wienholtz Theorem \cite{Ber68}, \cite{Gla65}, \cite{Wie58}. In \cite{ClaGes03}, Clark and Gesztesy extended the Glazman--Povzner--Wienholtz Theorem to the case of matrix-valued Schr\"odinger operators. The case of Hamiltonians \eqref{eq:hxa} with point interactions (and even with potentials that are $W^{-1,2}_{\loc}$-distributions) was treated in \cite{AKM_10}. Let us also mention the recent paper \cite{HryMyk12} by Hryniv and Mykytyuk, where an alternative proof of \cite[Theorem 1.1]{AKM_10} has  been proposed.
\end{remark}

\subsection{The form $\gt_{X,\gB,q}$} \label{ss:2.2}
We begin with some notation. Let  $q$ and $X$ be as in the previous subsection. Consider the following form in $L^2(\R_+)$
\begin{equation}\label{q_form}
\gq[f]:=\int_{\R_+}q(x)|f(x)|^2dx, \quad \dom (\gq) =\{f\in L^2(\R_+):\
|\gq[f]|<\infty\}.
\end{equation}
%
%
This form is  semibounded from below (and hence closed) if and only if so is $q$.

Next, define the  Hilbert space
    \begin{equation}\label{3.2pr}
 W^{1,2}(\R_+\setminus X) :=
\bigoplus^{\infty}_{k=1}W^{1,2}[x_{k-1},x_{k}],    \end{equation}
and introduce in $L^2(\R_+)$ the quadratic form
   \begin{align}\label{6.5B}
 &\gt_{X,q}[f] :=  \sum^{\infty}_{k=0}\int_{x_k}^{x_{k+1}}\big(|f'|^2 +q(x)|f|^2\big)dx  =  \int^{\infty}_0\big(|f'|^2 +q(x)|f|^2\big)dx,\\
 &\dom(\gt_{X,q}) = W^{1,2}(\R_+\setminus X;q):=\{f\in  W^{1,2}_{\loc}(\R_+\setminus X):\ |\gt_{X,q}[f] |<\infty\}.\label{6.5C}
    \end{align}
For $q=\bold{0}$, we set $\gt_{X}:=\gt_{X,\bold{0}}$ and $\dom(\gt_X)=W^{1,2}(\R_+\setminus X)$.

Clearly, with respect to the decomposition $L^2(\R_+)=\bigoplus_{k=1}^\infty L^2(x_{k-1},x_{k})$
the  form $\gt_{X,q}$ admits the representation
     \begin{align}\label{3.8A}
&\gt_{X,q} = \bigoplus^{\infty}_{k=1}\gt_{q,k}, \\
\gt_{q,k}[f] =
\int^{x_k}_{x_{k-1}}&\big(|f'|^2 +q(x)|f|^2\big)dt,\quad
\dom(\gt_{q,k})=W^{1,2}(x_{k-1},x_{k}).
     \end{align}
%
%
Since each form $\gt_{q,k}$ is closed and lower semibounded in $L^2(x_{k-1},x_k)$, the form $\gt_{X,q}$ is lower semibounded (and hence closed) if and only
if the forms $\gt_{q,k}$ have a finite uniform lower bound, i.e.,
there is $C>0$ such that
\[
\inf_{k\in\N}\inf_{f_k\in W^{1,2}(\Delta_k)}\gt_{q,k}[f_k]\ge -C\|f_k\|^2_{L^2[x_{k-1},x_{k}]}.
\]
In particular, the latter holds true if $q$ is lower semibounded on $\R_+$ (see also Corollary \ref{cor:3.5}).
Note also that
the operator associated with $\gt_{X,q}$ is $\rH_{X,q}^N$ given by \eqref{eq:H^N}.

The main object of this section is the following  form
    \begin{equation}\label{3.4pr}
\gt_{X,\beta,q}[f] := \int^{\infty}_0\bigl(|f'|^2 +
q(x)|f|^2\bigr)dx  + \sum^{\infty}_{k=1}\frac{|f
(x_k+)-f(x_k-)|^2}{\gB_k}.
    \end{equation}
 To define this form properly, firstly assume
\begin{equation}\label{hyp:b=0}
\beta_k\neq 0\quad \text{for all}\quad k\in\N.
\end{equation}
Now set
\[
\beta^\pm:=\{\beta^\pm_k\}_{k=1}^\infty,\quad \beta^\pm_k:=(|\beta_k|\pm \beta_k)/2.
\]
Also we split the sets $X$ and $\gB$ as follows
    \begin{equation}\label{eq:x_pm}
X_{\pm}:=\{x_k\in X:\   \pm\beta_k>0\}=\{x_k\in X: \gB_k^\pm\neq 0\},\quad \gB^\pm:=\{\gB_k^\pm:x_k\in X^\pm\}.
    \end{equation}
Moreover, we set
    \begin{equation}\label{eq:x_-}
X_-=\{x_j^-\}_{j=1}^\infty=\{x_{k_j}\}_{j=1}^\infty,\quad \gB_{k_j}=-\gB_{k_j}^-\, (<0),
    \end{equation}
where $
k_1<k_2<\dots<k_j<\dots$.

Next
define the (non-closed) perturbation forms 
${\mathfrak{b}}^{+}_{X}$ and $\gb_X^-$ by
    \begin{align}\label{6.5}
{\mathfrak{b}}^{\pm}_{X}[f]& :=
\sum_{x_k\in X_\pm}\frac{|f(x_k+)-f(x_k-)|^2}{\beta^{\pm}_k}=\sum_{x_k\in X}\Big(\frac{1}{\gB_k}\Big)^\pm|f(x_k+)-f(x_k-)|^2, \\
&\dom({\mathfrak{b}}^{\pm}_{X}) := \{ f\in
{W^{1,2}_{\loc}(\R_+\setminus X)}:\
|{\mathfrak{b}}^{\pm}_{X}[f]| <\infty\}.\nonumber
    \end{align}
Then we define ${\gt}_{X,\beta,q}$ as a sum of the forms defined
above,
    \begin{equation}\label{6.1B}
{\gt}_{X,\beta,q}:= {\gt}_{X,q} +
{\mathfrak{b}}^{+}_{X}  -  {\mathfrak{b}}^{-}_{X}, \quad
\dom(\mathfrak{t}_{X,\beta,q}) := {W}^{1,2}(\R_+\setminus X;q,
\beta),
       \end{equation}
where
     \begin{equation}\label{6.1BB}
{W}^{1,2}(\R_+\setminus X;q, \beta) :=  W^{1,2}(\R_+\setminus
X;q)\cap\dom({\mathfrak{b}}^{+}_{X})
\cap\dom({\mathfrak{b}}^{-}_{X}).
    \end{equation}
The form $\gt_{X,\beta,q}$ is not necessarily lower semi-bounded
and closed even if so is the form $\gt_{X,q}$.

The following simple result establishes a connection between the form $\gt_{X,\gB,q}$ and the operator $\rH_{X,\gB,q}$.
   \begin{proposition}\label{cor:2.4}
   Let $\rH^0_{X,\beta,q}$ be the minimal symmetric operator defined by
\eqref{eq:h_0}--\eqref{eq:dom_0}. Let also $\gt_{X,\beta,q}$ be the form defined by
\eqref{6.1B}--\eqref{6.1BB}. Then:
\begin{itemize}
\item[(i)] $\dom(\rH^0_{X,\gB,q})\subset \dom(\gt_{X,\gB,q})$ and
\begin{equation}\label{eq:3.18}
 (\rH^0_{X,\beta,q}f, f)_{L^2} = {\gt}_{X,\beta,q}[f],\quad f\in \dom(\rH^0_{X,\gB,q}).
\end{equation}
Moreover, $\dom(\rH^0_{X,\gB,q})$ is a core for  $\dom(\gt_{X,\gB,q})$.
\item[(ii)]
If the form $\gt_{X,\beta,q}$ is lower semibounded, then it is closable and the operator associated with its closure $\overline{\gt}_{X,\beta,q}$ coincides with the self-adjoint Hamiltonian
$\rH_{X,\beta,q}$.
\end{itemize}
\end{proposition}

\begin{proof}
(i) It is easy to see that $\dom(\rH^0_{X,\gB,q})\subset \dom(\gt_{X,\gB,q})$. Moreover, integrating by parts and using the fact that the function $f\in \dom(\rH^0_{X,\gB,q})$ satisfies
   \[
   f'(x_k^\pm)=\frac{f(x_k+)-f(x_k-)}{\gB_k},
   \]
   we obtain \eqref{eq:3.18}.

(ii) Firstly, (i) implies that the operator $\rH^0_{X,\gB,q}$ is lower semibounded. Therefore, by Theorem \ref{th_III.1}, it is essentially self-adjoint  and hence the form $\gt_{X,\gB,q}$ is closable. Moreover, by (i), $\dom(\rH^0_{X,\gB,q})$ is a core for $\gt_{X,\gB,q}$. This proves the second claim.
\end{proof}
   \begin{remark}
   Let us emphasize that Proposition \ref{cor:2.4} explains the difference between
   the Hamiltonians with $\delta$ and $\delta'$ interactions on $X$. Namely,
   the Hamiltonian with $\delta$-interactions is considered as a form perturbation of the
   free Hamiltonian in $L^2(\R_+)$ (cf. e.g., \cite{AKM_10}). However, the Hamiltonian $\rH_{X,\gB,q}$
   is a form perturbation of $\rH_{X,q}^N$ (see \eqref{eq:H^N}), which is the direct sum of
   Neumann realizations $\rH_{q,k}^N$. This fact explains a substantial difference between  the spectral properties
   of Hamiltonians $\rH_{X,\gA,q}$ and $\rH_{X,\gB,q}$ with $\delta$ and $\delta'$
 interactions, respectively.
   \end{remark}
%
%


   \subsection{Necessary and sufficient conditions for $\gt_{X,\gB,q}$ to be closed}\label{ss:2.3}
  Let the form ${\gt}_{X,\beta,q}$ be lower semibounded, ${\gt}_{X,\beta,q}\ge -c$. Hence it is closable. Denote its closure by ${\overline{\gt}}_{X,\beta,q}$. Let also $\gH_{X,\gB,q}$ be the Hilbert space naturally associated  with $\overline{\gt}_{X,\beta,q}$, i.e.,  $\gH_{X,\gB,q}=\dom({\overline{\gt}}_{X,\beta,q})$  equipped with the energy norm
   \begin{equation}\label{3.17A}
\|f\|^2_{{\gH}_{X,\beta,q}} = {\overline{\gt}}_{X,\beta,q}[f] +
(1+c)\|f\|^2_{L^2}, \qquad f\in \dom({\overline{\gt}}_{X,\beta,q}).
   \end{equation}

In this section, we are going to indicate some cases when the semibounded form $\gt_{X,\gB,q}$ is closed
and then to describe the domain of  $\gt_{X,\gB,q}$ in two important cases.
%
%
%
%
\begin{lemma}\label{lem_III.3}
Assume that $q$ is lower semibounded, $q\geq -c$ a.e. on $\R_+$. Let the forms $\gt_{X,q}$
and ${\gb}^+_{X}$ be defined by \eqref{6.5B}--\eqref{6.5C} and \eqref{6.5},
respectively. Then the form $\gt_{X, {\gB}^+,q}={\gt}_{X,q} +
{\gb}^+_{X}$ defined by \eqref{6.1B} is semibounded and closed.  Moreover,
${\gH}^+_{X,\beta,q} = W^{1,2}(\R_+\setminus X; \beta^+,q)$
algebraically and topologically and the operator associated with
$\gt_{X, {\gB}^+,q}$ is $\rH_{X, {\gB}^+,q} =
\rH_{X,{\gB}^+,q}^*.$
    \end{lemma}
     \begin{proof}
Without loss of generality we can assume that $q$ is nonnegative, $c=0$. Consider the Hilbert space ${\gH}^+_{X,\beta,q}$ equipped  with the norm \eqref{3.17A}.
Let $\{f_n\}_{n=1}^\infty$ be a Cauchy sequence in $\gH^+_{X,\gB,q}$.
Since ${ W}^{1,2}(\R_+\setminus X;q)$  and
$\ell^2(\N;\{(\beta^+)^{-1}\})$ are Hilbert spaces, there exist $f\in
{W}^{1,2}(\R_+\setminus X;q)$ and
 $y=\{y_k\}^{\infty}_1\in  \ell^2(\N;\{(\beta^+)^{-1}\})$ such that
$\lim_{n\to\infty}\|f_n-f\|_{{W}^{1,2}(\R_+\setminus X;q)}=0$  and
$\lim_{n\to\infty}\sum_k(\beta_k^+)^{-1}|f_n(x_k + )- f_n(x_k - )
-y_k|^2 = 0$. However, the Sobolev space $W^{1,2}[x_k,x_{k+1}]$ is
continuously embedded into $C[x_k, x_{k+1}]$. Therefore, for all
$x\in X$ we get
 $\lim_{n\to\infty}f_n(x_k\pm)=f(x_k\pm)$ and hence $y_k = f(x_k+) - f(x_k-)$. The latter yields
 $f\in {\gH}^+_{X,\gB,q}$ and $\lim_{n\to\infty}\|f_n - f\|_{{\gH}^+_{X,\gB,q}} = 0$.
Thus, ${\gH}^+_{X,\gB,q}$ is a Hilbert space and hence the form  $\gt_{X,\gB^+,q}$ is closed.
\end{proof}


Before proceeding  further, we  need the following simple but
useful fact. 
   \begin{lemma}\label{lem_III.1}
Assume that $\gd^*=\sup_k\gd_k<\infty$. Assume also that 
\begin{equation}\label{III.6}
C_0:=\sup_{k\in\N} \frac{1}{\gd_k}\int_{\Delta_k}|q(x)|dx<\infty,\qquad
C_1:=\sup_{k\in\N} \frac{|\gB_k|^{-1}}{\min\{\gd_k,\gd_{k+1}\}}<\infty.
  \end{equation}
%
Then:
\begin{itemize}
\item[(i)] the forms $\gq$, $\gb_X^+$, $\gb_X^-$ and $\gb_{X}:=\gb^+_{X} + \gb^-_{X}$ are
infinitesimally $\gt_X$-bounded,
\item[(ii)]  the form
$\gt_{X,\beta,q}$ is lower semibounded and closed and
$\gH_{X,\gB,q} = W^{1,2}(\R_+\setminus X)$ algebraically and
topologically.
\end{itemize}
  \end{lemma}
\begin{proof}
(i) By the Sobolev embedding theorem,  the following inequality (cf. \cite[inequality (IV.1.19)]{Kato66})
    \begin{equation}\label{3.3}
|f(x)|^2\le \varepsilon\gd_k \int_{\Delta_k}|f'(t)|^2dt +
\frac{C}{\varepsilon\gd_k} \int_{\Delta_k}|f(t)|^2dt \le \varepsilon\gd_k
\|f\|^2_{W^{1,2}(\Delta_k)} +
\frac{C}{\varepsilon\gd_k}\|f\|^2_{L^2({\Delta_k})},
    \end{equation}
holds for all $x\in[x_{k-1},x_k]$ and  $\varepsilon>0$ with some constant $C>0$ independent of $f$ and $k\in \N$. Firstly, let us estimate $\gq[f]$. Using
\eqref{III.6} with \eqref{3.3}, we obtain for $f\in
\dom(\gt_X) =  W^{1,2}(\R_+\setminus X)$
\begin{align*}
|\gq[f]|&\le\sum_{k=1}^\infty\int_{\Delta_k}|q(x)||f(x)|^2dx \le \sum_{k=1}^\infty\|f\|^2_{ C(\Delta_k)}\int_{\Delta_k}|q(x)|dx   \\
&\le \sum^{\infty}_{k=1} \gd_kC_0\|f\|^2_{ C(\Delta_k)} \le  \varepsilon {\gd^*}^2C_0\|f\|^2_{W^{1,2}(\R_+\setminus X)} + \frac{CC_0}{\varepsilon}
\|f\|^2_{L^2(\R_+)}.
    \end{align*}
    Similarly, we estimate the form $\gb_X$:
    \begin{align*}
|\gb_X[f]|&\le \gb_X^+[f]+\gb_X^-[f]\le\sum^{\infty}_{k=1}\frac{|f(x_k+)-f(x_k-)|^2}{|\beta_k|} \le 2\sum^{\infty}_{k=1  }\frac{\|f\|^2_{ C(\Delta_k)}+\|f\|^2_{ C(\Delta_{k+1})}}{|\beta_k|}  \\
& \le
2\sum^{\infty}_{k=1  }\Big(C_1\gd_k \|f\|^2_{ C(\Delta_k)}+C_1\gd_{k+1}\|f\|^2_{ C(\Delta_{k+1})}\Big)\\
&\le 4C_1\sum^{\infty}_{k=1} \gd_k\|f\|^2_{ C(\Delta_k)} \le 4C_1\varepsilon  {\gd^*}^2\|f\|^2_{W^{1,2}(\R_+\setminus X)} + \frac{4CC_1}{\varepsilon}
\|f\|^2_{L^2(\R_+)}.
    \end{align*}
Since $\varepsilon>0$ is arbitrary, the forms $\gq$, $\gb_X^+$, $\gb_X^-$ and
$\gb_{X}$ are infinitesimally form bounded with respect to
$\gt_X$.

(ii) immediately follows from the KLMN theorem (Theorem \ref{th_KLMN}).
\end{proof}
      \begin{proposition}\label{cor:3.3}
Assume that $d^*<\infty$ and  $q_-$ and $\gB^-$ satisfy \eqref{III.6}
, i.e.,
      \begin{equation}\label{eq:b_klmn}
\frac{1}{\gd_k}\int_{\Delta_k}q_-(x)dx\le C_0, \qquad
\Big(\frac{1}{\gB_k}\Big)^-\le C_1 \min\{\gd_{k},\gd_{k+1}\},\quad k\in\N,
   \end{equation}
with some constants $C_0,C_1>0$. Then:
\begin{itemize}
\item[(i)] the forms $\gq_-$ and
$\gb_X^-$ are infinitesimally $\gt_{X}$ bounded and hence the form
$\gt_{X,\gB,q}$ is lower semibounded and closed,
\item[(ii)] the following equalities
  \begin{equation}\label{eq:dom=w}
\gH_{X,\gB,q} = W^{1,2}(\R_+\setminus X; \gB,q) =
W^{1,2}(\R_+\setminus X; \gB^+,q_+) = \gH_{X,\gB^+,q_+}
  \end{equation}
hold algebraically and topologically, and the operator associated with
$\gt_{X,\gB,q}$ is $\rH_{X,\gB,q} = \rH_{X,\gB,q}^*$,     
\item[(iii)]
 if, additionally, $q_+$ and $\gB^+$ satisfy \eqref{III.6}, then  \eqref{eq:b_klmn}
is also necessary for the form $\gt_{X,\gB,q}$ to be lower semibounded.  In particular, conditions \eqref{eq:b_klmn} are necessary for lower semiboundedness  whenever $q$ and $\gB$ are negative.
\end{itemize}
     \end{proposition}
      \begin{proof}
(i) and (ii) immediately follow from Lemmas \ref{lem_III.3} and \ref{lem_III.1}.

To prove (iii), set $h_k(x)=\frac{1}{\sqrt{\gd_k}}\chi_{\Delta_k}(x)$. Noting that $\|h_k\|_{L^2}=1$ and $h_k\in W^{1,2}(\R_+\setminus X)$, we obtain
\begin{align*}
\gt_{X,\gB,q}[h_k]= \frac{1}{\gd_k}\int_{\Delta_k}q(x)dx+\frac{1}{\gB_{k-1}\gd_k}+\frac{1}{\gB_{k}\gd_k}\ge -C\|h_k\|^2_{L^2}=-C, \quad k\in\N.
\end{align*}
Therefore,
\[
-\frac{1}{\gd_k}\int_{\Delta_k}q_-(x)dx-\frac{1}{\gB^-_{k-1}\gd_k}-\frac{1}{\gB^-_{k}\gd_k}\ge -C-C_0-C_1=: -\widetilde{C}, \qquad k\in\N.
\]
where $C_0$ and $C_1$ are given by \eqref{III.6} with $q_+$ and $\gB^+$  in place of $q$ and $\gB$, respectively.
Since all summands in the left-hand side of this inequality are negative, we get
\[
\frac{1}{\gd_k}\int_{\Delta_k}q_-(x)dx\le \widetilde{C},\qquad \frac{1}{\gB^-_{k-1}}\le \widetilde{C}\gd_k,\qquad \frac{1}{\gB^-_{k}}\le \widetilde{C}\gd_k,\qquad k\in\N.
\]
This implies \eqref{eq:b_klmn}.
\end{proof}
 \begin{corollary}\label{cor:3.5}
Let $q_-$ satisfy the first condition in \eqref{eq:b_klmn}. Then the
the form $\gt_{X,q}$ (and hence the operator $\rH_{X,q}^N = \bigoplus^{\infty}_{k=1}\rH^N_{q,k}$ given by \eqref{eq:H^N}) is
lower semibounded and $\gH_{X,q} = \gH_{X,q_+} =
W^{1,2}(\R_+\setminus X; q_+)$ algebraically and topologically. If additionally $q=q_-$, then the
first condition in \eqref{eq:b_klmn} is also necessary for lower
semiboundedness of $\gt_{X,q}$ (and hence $\rH_{X,q}^N$).
   \end{corollary}
\begin{proof}
The proof is analogous to that of Proposition \ref{cor:3.3} and we left it to the reader.
\end{proof}
\begin{remark}\label{rem:2.7}
Note that (see, e.g., \cite{AKM_10}) the perturbation of the free Hamiltonian
\[
\rH_0=-\frac{d^2}{dx^2},\quad \dom(\rH_0)=\{f\in  W^{2,2}(\R_+):\, f'(0)=0\},
\]
by a negative potential $q=-q_-$ is lower semibounded if and only if $q_-$ satisfies
   \begin{equation}\label{eq:brinck}
\sup_{x\in\R_+}\int_{[x,x+1]}q_-(t)dt <+\infty. 
\end{equation}
Clearly, condition \eqref{eq:b_klmn} implies \eqref{eq:h_0} if $\gd^*<\infty$. Indeed,
\[
\int_{x}^{x+1}q_-(x)dx\le \sum_{k:\, \Delta_k\cap [x,x+1]\neq \emptyset}\int_{\Delta_k}q_-(x)dx\le\sum_{k:\, \Delta_k\cap [x,x+1]\neq \emptyset}C_0\gd_k\le C_0(1+2\gd^*).
\]

However, the converse is not true. Indeed, set $X=\{x_k\}_1^\infty$, where
$x_{2k-1}=k$, $x_{2k}=k+1/2k$, $k\in\N$. Define $q$ as follows:
\be\label{eq:ex2.1}
q(x)=-\sum_{k=1}^\infty k\chi_{\Delta_{2k}}(x).
\ee
Evidently, $q$ satisfies \eqref{eq:brinck} with $C=1/2$, however,
\[
\frac{1}{\gd_{2k}}\int_{\Delta_{2k}}q_-(x)dx=k\int_{[k,k+1/k]}k\, dx=k,
\]
 and hence $q$ does not satisfy \eqref{eq:b_klmn}. The latter, in particular,
 gives an explicit example of the Hamiltonian $\rH_q=-d^2/dx^2+q(x)$  
 which is lower semibounded in $L^2(\R_+)$, although  the Neumann realization $\rH_{X,q}^N$ defined by \eqref{eq:H^N}
 is not bounded from below. 
    \end{remark}

\begin{remark}\label{rem3.9}
 $(i)$ If $\gB^-$ satisfies  \eqref{eq:b_klmn}, then, by Proposition \ref{cor:3.3}(ii), we get the following implication 
\[
f\in \dom(\rH_{X,\gB}) (\subset \dom(\gt_{X,\gB}))\quad \Longrightarrow \quad f\in \dom(\gt_{X,\gB^+}) = \dom(\gt_{X,\gB})\subseteq
\mathop{W^1_2}(\R_+\setminus{X}).
\]

However, this implication may be false if $\gt_{X,\gB}$ is semibounded but $\gB^-$ does not satisfy \eqref{eq:b_klmn}.  Let us mention that the latter also does not hold for Schr\"odinger operators with locally integrable potentials (see \cite{Eve83}, \cite{EveGieWei73}, \cite{Kal74} and references therein).

$(ii)$ Assume that the form $\gt_{X,\gB,q}$ is lower semibounded.
By Corollary \ref{cor:2.4},  it is closable and let
$\overline{\gt}_{X,\gB,q}$ be its closure. If the condition
\eqref{eq:b_klmn} is not satisfied, then  it might happen that   
$f\in \dom(\gt_{X,q}) \cap \dom(\overline{\gt}_{X,\gB,q})$ but
$f\notin \dom(\gb^\pm_{X})$ (cf., \cite{Bri59} and \cite[Example 2]{Bra85}).

$(iii)$ I. Brinck \cite{Bri59} proved that the Hamiltonian $\rH_{q}$ is lower semibounded if there exists a constant $C\ge 0$ such that
\be\label{eq:bri59}
\int_{x}^{x+\varepsilon}q(t)dt\ge -C\quad\text{for all}\ \ x>0 \ \ \text{and}\ \ \varepsilon\in(0,1).
\ee
Clearly, if the negative part $q_-$ of $q$ satisfies \eqref{eq:brinck}, then $q$ satisfies \eqref{eq:bri59}. However, the converse is not true (take $q(x)=x^n\sin(x^{n+1})$, \cite{Bri59}). Moreover, it is shown in \cite{Bri59} that $\dom(\rH_q)\subset W^{1,2}(\R_+)$ if \eqref{eq:bri59} holds. However,  Brinck's condition \eqref{eq:bri59} does not imply the algebraic (and topologic) equality $\gH_q = W^{1,2}(\R_+;q_+)$, i.e., the energy space $\frak H_q$ can be wider than $W^{1,2}(\R_+;q_+)$. For instance, for $q(x)=x^n\sin(x^{n+1})$ and $f\in \dom(\rH_q)$ the integral $\int_{\R_+}q(x)|f|^2dx$ might be infinite,  although  the following limit
\[
\lim_{N\to+\infty}\int_0^Nq(x)|f(x)|^2dx
\]
exists and is finite for every $f\in\dom(\rH_q)$. Using this example one can construct a lower semibounded Hamiltonian  $H_{X,\gB,q}$  such that  the  corresponding energy space
$\gH_{X,\gB,q}$ is  wider than $W^{1,2}(\R_+\setminus X;\gB,q)$.
     \end{remark}

\subsection{The case $q\in L^\infty(\R_+)$}\label{ss:2.4}
Now we restrict our considerations to the case of a bounded potential $q$. More precisely, we assume that $q\equiv 0$. Our main aim is to give several  simple necessary and sufficient conditions for $\gt_{X,\gB}$ to be lower semibounded as well as to provide some estimates for the lower bound in terms of $\gB$ and $X$.

\begin{lemma}\label{lem:nec}
Let $X_-=\{x_j^-\}_{1}^\infty=\{x_{k_j}\}_1^\infty$ be the set supporting negative intensities and defined by \eqref{eq:x_pm}--\eqref{eq:x_-}.
If the form  $\gt_{X,\gB}$ is lower semibounded, that is $\gt_{X,\gB}\ge -C$ for some $C\ge 0$, then:
\begin{enumerate}
\item[(i)] for all negative $\gB_k$
  \begin{equation}\label{eq:inf_b}
\Big(\frac{1}{\gB_k}\Big)^-\le 1+\frac{C}{3},\quad k\in\N,
  \end{equation}
\item[(ii)]
  \begin{equation}\label{eq:nec_1}
\frac{1}{\gB_j^-}=\frac{1}{|\gB_{k_j}|}\le C\min \{\gd_j^-,\gd_{j+1}^-\},\qquad 
 \quad j\in\N,
  \end{equation}
where $\gd_j^-:=x_j^--x_{j-1}^-=x_{k_j}-x_{k_{j-1}}$,
  \item[(iii)]
for all $i\in\{1,\dots,k_{j+1}-k_j\}$
  \begin{equation}\label{eq:nec_2}
\frac{1}{\gB_{k_j}}+\frac{1}{\gB_{k_j+i}}\ge -C(x_{k_j+i}-x_{k_j}),
  \end{equation}
  \item[(iv)] for all $i\in\{1,\dots,k_{j}-k_{j-1}\}$
  \begin{equation}\label{eq:nec_3}
\frac{1}{\gB_{k_j-i}}+\frac{1}{\gB_{k_j}}\ge -C(x_{k_j}-x_{k_j-i}),
   \end{equation}
%
%
%
\end{enumerate}
\end{lemma}
\begin{proof}
(i) Set
\[
f_k(x)=
\begin{cases}
x_k+1-x,&  x\in[x_k,x_k+1],\\
0,& x\notin [x_k,x_k+1]
\end{cases}
\]
Clearly, $\|f_k\|^2_{L_2}=1/3$, $\|f'_k\|^2=1$ and $\gb_X[f_k]=\gB_k^{-1}$. Therefore,  the inequality $\gt_{X,\gB}[f_k]\ge -C\|f_k\|^2_{L^2}$ implies \eqref{eq:inf_b}.

(ii)  Now set $f_{k_j}=\chi_{[x_{k_j},x_{k_{j+1}}]}$. Then $\gt_{X,\gB}[f_k]\ge -C\|f_k\|^2_{L^2}$ reads
as follows
\[
\frac{1}{\gB_{k_j}}+\frac{1}{\gB_{k_{j+1}}}\ge -C(x_{k_{j+1}}-x_{k_j})=-C\gd_{j+1}^-.
\]
Noting that $\gB_{k_j}$ is negative for all $j\in\N$, one infers
\[
\frac{1}{|\gB_{k_j}|}\le C\gd_j^-,\quad \frac{1}{|\gB_{k_j}|}\le C\gd_{j+1}^-,
\]
which completes the proof of \eqref{eq:nec_1}. 

(iii)-(iv)  Setting $f_{k_j+i}=\chi_{[x_{k_j},x_{k_j+i}]}$, it
easy to check that $\gt_{X,\gB}[f_{k_j+i}]\ge -C\|f_{k_j+i}\|^2$
is equivalent to the estimate in  \eqref{eq:nec_2}. Similarly, the
substitution $f_{k_j-i}:=\chi_{[x_{k_j-i},x_{k_j}]}$ proves
\eqref{eq:nec_3}.
%
%
%
%
%
%
%
%
%
%
%
%
\end{proof}
The next simple example shows that conditions \eqref{eq:inf_b} and \eqref{eq:nec_1} are only necessary.
\begin{example}
Let $X=\{x_k\}_1^\infty$ and $x_{2k-1}=k$, $x_{2k}=k+\frac{1}{2k}$, $k\in\N$. Set
\[
\gB_k=\begin{cases}
-1, & k=2j-1\\
1/j, & k=2j
\end{cases}.
\]
Clearly, the sequence $\gB$ is bounded and hence satisfies \eqref{eq:inf_b}. Further, note that $\gB$ satisfies \eqref{eq:nec_1}. Indeed, $x_j^-=2j-1$ and $\gd_j^-=1$, $j\in\N$. Since $\gB_{2j-1}^-=1$, we see that $\gB$ satisfies \eqref{eq:nec_1} with $C=1$.

However,
\[
\frac{1}{\gB_{2j-1}}+\frac{1}{\gB_{2j}}=-1+\frac{1}{j},\quad \gd_{2j}=\frac{1}{j},\quad j\in\N,
\]
and hence $\gB$ does not satisfy \eqref{eq:nec_2}. Thus the corresponding form $\gt_{X,\gB,q}$ is unbounded from below.
\end{example}

The next results demonstrates that under additional assumptions on $X$ the condition \eqref{eq:nec_1} is equivalent to \eqref{eq:b_klmn} and hence also necessary for semiboundedness.
   \begin{corollary}\label{cor3.4}
Assume that $d^*<\infty$. If there is $C_2>0$ such that
         \begin{equation}\label{3.32pr}
\gd_j^-:=d_{k_{j-1}+1} + \ldots + d_{k_j}  \le C_2\min\{\gd_{k_{j-1}+1}, \gd_{k_j}\},\quad j\in\N,
         \end{equation}
then condition \eqref{eq:nec_1} is necessary and sufficient for
$\rH_{X,\beta}$ to be lower semibounded.
    \end{corollary}
       \begin{proof}
By Lemma \ref{lem:nec}, it remains to prove that \eqref{eq:nec_1} is sufficient. However,
\[
\frac{1}{\gB^-_{k_j}}\le C\min\{\gd_j^-,\gd_{j+1}^-\}\le CC_2 \min\{\gd_{k_{j-1}+1},\gd_{k_j},\gd_{k_{j}+1},\gd_{k_{j+1}}\}\le CC_2\min\{\gd_{k_j},\gd_{k_{j}+1}\}.
\]
Proposition \ref{cor:3.3} completes the proof.
    \end{proof}

Next  we indicate simple additional conditions that allow to
obtain  criteria of lower semi-boundedness. The first criterion
depends on a geometry of $X$ and reads as follows.
       \begin{corollary}[\cite{Mih_96a}]
 Let  $0<d_*<d^*<\infty$, then condition
       \begin{equation}\label{3.37prA}
\sup_{k\in\N} \Big(\frac{1}{\beta_k}\Big)^- <\infty
           \end{equation}
is necessary and sufficient for $\rH_{X,\beta}$ to be semibounded
below.
    \end{corollary}
    \begin{proof}
 \emph{Necessity}  of condition \eqref{3.37prA}  was established in Lemma \ref{lem:nec}(i).
\emph{Sufficiency} is implied by Corollary \ref{cor:3.3}.
    \end{proof}


Next we show that the semi-boundedness of the Hamiltonian
$\rH_{X,\beta}$ yields  some natural  restrictions on the negative
part of intensities, which are stronger than the boundedness condition \eqref{eq:inf_b}.
         \begin{corollary}\label{cor3.7pr}
Let $d^*<\infty$. Assume also that $\sup_j(k_{j+1}-k_j) = K < \infty$    and  $\lim_{k\to\infty}d_k=0$.
If the form $\gt_{X,\beta}$ is  lower semi-bounded,
$\gt_{X,\beta}\ge - C $,  then
    \begin{equation}\label{3.36pr}
\lim_{j\to\infty}|\beta^- _j|^{-1} =  \lim_{j\to\infty}|\beta
_{k_j}|^{-1} = 0.
    \end{equation}
      \end{corollary}
     \begin{proof}
Since $\lim_{k\to\infty}d_k=0,$  for any $\varepsilon
> 0$ there exists $N = N(\varepsilon)$ such that $d_k <
\varepsilon$ for $k>N.$  Therefore inequality \eqref{eq:nec_1}
yields $(\beta_j^-)^{-1}\le C_2(d_{k_{j-1}+1} + \ldots + d_{k_j})
\le C_2 K\varepsilon$ and hence \eqref{3.36pr} follows.
         \end{proof}
%
%
%
%
%
%
     \begin{remark}\label{rem3.10}
Note that self-adjointness of not necessarily semibounded  Hamiltonians
$\rH_{X,\gA,q}$ has been investigated in several papers
\cite{Alb_Ges_88}, \cite{BusStoWei95}, \cite{KosMal09}, \cite{KosMal10}, \cite{KosMal12} (see also references therein).
 Let us mention that in the case $q\in L^\infty(\R_+)$ conditions for semiboundedness of $\rH_{X,\gB}$ similar to \eqref{eq:b_klmn} and \eqref{eq:inf_b} have been obtained in \cite{KosMal09}, \cite{KosMal10} by using a different approach.

    \end{remark}

\section{Operators with discrete spectrum}\label{sec:III}

Recall that according to the classical result of A.M. Molchanov
\cite{Mol53}, \cite{Gla65} (see also \cite{AKM_10}, where Hamiltonians with $\delta$-interactions have been considered), {\em the Sturm--Liouville operator
$\rH_q=-\frac{d^2}{dx^2} +q(x)$  with a lower semibounded potential $q \ge -c$ has discrete spectrum if and only if  }
    \begin{equation}\label{4.1}
\lim_{x\to\infty}\int^{x + \varepsilon}_x q(t)dt=+ \infty\quad \text{for every}\ \varepsilon>0. 
    \end{equation}
%
%

 Here we prove necessary and sufficient conditions, which are in a certain sense analogous to the
Molchanov theorem. In particular, we shall show that Molchanov's condition \eqref{4.1} remains to be necessary for the discreteness. 
However, it is no longer sufficient. Namely, we emphasize that for Hamiltonians $\rH_{X,\gB,q}$ a new additional condition
\eqref{Intro2.2}  appears.
%
%
%
%

\subsection{Necessary conditions}\label{ss:3.1}
We begin with the following result.

        \begin{proposition}\label{prop4.5}
Assume that  $q$ satisfies  condition \eqref{eq:brinck}, that is
\be\label{eq:4.2}
\sup_{x\in\R_+}\int_{[x,x+1]}q_-(t)dt <+\infty.
\ee 
 Assume also that $q$ does not satisfy Molchanov's condition \eqref{4.1}. If the operator $\rH_{X,\beta,q}$  is lower semibounded, then it is self-adjoint and its spectrum is  not
discrete.
    \end{proposition}
     \begin{proof}
By Theorem \ref{th_III.1},  $\rH_{X,\beta,q}$ is self-adjoint, $\rH_{X,\beta,q} =
\rH_{X,\beta,q}^*\ge -c$.   Let $\gt_{X,\beta, q}$ be the corresponding
quadratic form \eqref{3.4pr},
    \begin{equation}\label{4.18}
\frak t_{X,\beta, q}[f]=  \int^{\infty}_0 \left(|f'(x)|^2  + q(x)|f(x)|^2\right) dx +
\sum^{\infty}_{k=1}\frac{|f (x_k+)-f(x_k-)|^2}{\beta_k},
    \end{equation}
and let  $\gH_{X,\beta, q}$ be the corresponding Hilbert space, i.e., the closure of $\dom(\frak t_{X,\beta, q})$ equipped with the energy norm  \eqref{3.17A}.

Further, let $\rH_q$ be the Neumann realization of  $\tau_q=-d^2/dx^2 + q$
in $L^2(\R_+)$. Since  $q$ satisfies  \eqref{eq:4.2}, the operator $\rH_q$ is self-adjoint and
lower  semibounded, $\rH_q=\rH_q^*\ge -c_1I$ (see, e.g., \cite{AKM_10}), and the corresponding form 
\[
\gt_{q}[f] = \int_0^\infty \left(|f'(x)|^2  + q(x)|f(x)|^2\right)\,dx, \quad \dom(\gt_{q}) = W^{1,2}(\R_+; q),
 \]
is well defined and closed.  Moreover, the corresponding energy space $\gH_{q}$ coincides with $W^{1,2}(\R_+; q)$  algebraically  and topologically.
Note  also that
\begin{equation}\label{4.33}
\gH_{q} = W^{1,2}(\R_+;q)\subset \dom(\overline{\gt}_{X,\beta, q})=  \gH_{X,\beta, q} \quad \text{and}\quad
\gt_{X,\beta, q}[f] = \gt_{q}[f]
\end{equation}
for  all $f\in W^{1,2}(\R_+; q)$.

Let us show that the embedding \eqref{4.33} is continuous.
By the second representation theorem,  $\gH_{q}$ and $\gH_{X,\beta,q}$ coincide algebraically
and topologically with $\dom(\rH_{q} +c_1I)^{1/2}$ and $\dom(\rH_{X,\beta,q}+cI)^{1/2}$, respectively, which  are
equipped with the corresponding graph norms.
Therefore, by \cite[Theorem 2.6.2]{Ios65} (see also \cite[Remark IV.1.5]{Kato66}), the embedding $i_1: W^{1,2}(\R_+;q)= \gH_{q} \hookrightarrow \gH_{X,\beta,q}$ (see  \eqref{4.33}) is continuous.

Finally, if the spectrum  $\sigma(\rH_{X,\beta,q})$ is discrete, then, by
the Rellieh theorem, the embedding $i_2: \gH_{X,\beta,q}
 \hookrightarrow L^2(\R_+)$ is compact. Taking the composition
 $i = i_2 i_1:  W^{1,2}(\R_+;q)\to L^2(\R_+)$ one gets that
the embedding $i:W^{1,2}(\R_+;q)\hookrightarrow L^2(\R_+)$ is compact. Now Molchanov's theorem implies that condition \eqref{4.1} is satisfied. This contradiction completes the proof.
     \end{proof}
As an immediate corollary of Proposition \ref{prop4.5} we obtain the following result.
        \begin{corollary}\label{cor3.8}
Let $q\in L^\infty(\R_+)$. If the (self-adjoint) operator $\rH_{X,\beta,q}$ is lower semibounded, then its spectrum is not discrete.
In particular, $\rH_{X,\beta} = \rH_{X,\beta,0}$ is not discrete whenever it is lower semibounded.
    \end{corollary}
\begin{proof}
If $q\in L^\infty(\R_+)$, then it satisfies \eqref{eq:4.2} and does not satisfy \eqref{4.1}. Proposition \ref{prop4.5} completes the proof.
  \end{proof}
\begin{proposition}\label{prop:4.3}
If the operator $\rH_{X,\gB,q}$ is lower semibounded and has discrete spectrum, then
\be\label{eq:4.5'}
\frac{1}{\gd_k}\Big(\int_{x_{k-1}}^{x_k}q(x)dx+\frac{1}{\gB_{k-1}}+\frac{1}{\gB_k}\Big)\to +\infty.
\ee
\end{proposition}

\begin{proof}
%
%
%
%
Consider the form $\gt_{X,\gB,q}$ given by \eqref{3.4pr}. Note that it is lower semibounded, $\gt_{X,\gB,q}\ge -c$, and closable in $L^2(\R_+)$ since the operator $\rH_{X,\gB,q}$ is lower semibounded.

Set
   \begin{equation}\label{4.9A}
h_{k}:= d^{-1/2}_{k}\chi_{\Delta_k}(x)=
\begin{cases}
d^{-1/2}_{k}, & x\in[x_{k-1},x_k] \\
0,& x\notin [x_{k-1},x_k]
\end{cases},\quad k\in\N.
\end{equation}
%
%
Clearly, $\|h_k\|_{L^2(\R_+)}=1$ and $h_{k}\in\dom(\gt_{X,\gB,q})$. Therefore, using \eqref{4.18}, we get
   \begin{equation}\label{eq:4.8B}
\gt_{X,\beta,q}[h_{k}] = \frac{1}{\gd_{k}}
\Big(\int_{\Delta_k}q(x) dx + \big(\gB_{{k}-1}^{-1} +
\gB_{k}^{-1}\big) \Big), \quad k\in\N.
   \end{equation}

Noting that the form $\gt_{X,\gB,q}$ is lower semibounded, we see that
\[
\inf_k\gt_{X,\beta,q}[h_{k}]\ge -c>-\infty.
\]

Assume that \eqref{eq:4.5'} does not hold, that is, there is a subsequence $\{h_{k_j}\}_{j=1}^\infty$ such that
\[
\gt_{X,\beta,q}[h_{k_j}]\le C_0<\infty.
\]
The latter also means that the subsequence $h_{k_j}$ is bounded in $\gH_{X,\gB,q}$.
On the other hand, the system $\{h_{k}\}^{\infty}_{k=1}$ is
orthonormal in $L^2(\R_+)$ and hence is not compact there. Therefore, the embedding $\gH_{X,\gB,q}\hookrightarrow L^2(\R_+)$ is not compact and hence, by Theorem \ref{th_Rel}, the spectrum of the operator $\rH_{X,\gB,q}$ is not compact. This contradiction completes the proof.
\end{proof}
Now we are ready to prove that conditions \eqref{4.1} and \eqref{eq:4.5'} are necessary for the discreteness.

    \begin{proof}[Proof of Theorem \ref{th_discretcriter}. Necessity.]
    Assume that $q$ satisfies \eqref{I_brinck}, that is,
    \be\label{eq:4.8}
    \sup_k\frac{1}{\gd_k}\int_{\Delta_k}q_-(x)dx<+\infty.
    \ee
    Then (see Remark \ref{rem:2.7}) $q$ satisfies \eqref{eq:4.2} and hence, by Proposition \ref{prop4.5}, \eqref{4.1} is necessary.
    Noting that condition \eqref{eq:4.5'} is necessary by Proposition \ref{prop:4.3}, we complete the proof.
    \end{proof}

  \subsection{Discreteness of the spectrum of $\rH_{X,q}^N$}  \label{ss:3.2}
  The main result of this subsection is the following discreteness criterion for the operator $\rH_{X,q}^N$ defined by \eqref{eq:H^N}--\eqref{eq:HkN2}.
  %
  %
  \begin{theorem}\label{th:discr_N}
  Let $q\in L^1_{\loc}(\R_+)$ and $\gd^*<\infty$. Assume that the operator $\rH_{X,q}^N$ given by \eqref{eq:H^N}--\eqref{eq:HkN2} is lower semibounded. If the potential $q$ satisfies \eqref{eq:4.8}, then the spectrum of $\rH_{X,q}^N$ is purely discrete if and only if $q$ satisfies
  \begin{itemize}
  \item[(i)]  Molchanov's condition \eqref{4.1} and
  \item[(ii)] \be\label{eq:4.9}
  \frac{1}{\gd_k}\int_{\Delta_k}q(x)dx\to +\infty.
  \ee
  \end{itemize}
  \end{theorem}
  \begin{proof}
  The proof of necessity of conditions \eqref{4.1} and \eqref{eq:4.9} is analogous to that of Proposition \ref{prop4.5} and \ref{prop:4.3} and we left it to the reader.

  Let us prove sufficiency. Firstly, note that, by Corollary \ref{cor:3.5}, the form $\gt_{X,q}$ given by \eqref{6.5B}--\eqref{6.5C} is lower
 semibounded and closed in $ L^2(\R_+)$.
 Moreover, the corresponding energy space is  $\gH_{X,q} =
 W^{1,2}(\R_+\setminus X; q)=W^{1,2}(\R_+\setminus X; q_+)$. The latter holds algebraically and topologically. 
Therefore, it suffices to prove sufficiency for nonnegative potentials.  Hence
without loss of generality we can assume that $q\ge 1$. 

By the Rellieh theorem (Theorem \ref{th_Rel}), it suffices to show that the embedding $i:
W^{1,2}(\R_+\setminus X; q)\hookrightarrow L^2(\R_+)$ is
compact, i.e.,  the unit ball  
%
%
     \begin{equation}\label{4.2}
\mathbb{U}_{X, q}:= \{f\in W^{1,2}(\R_+\setminus X; q): \ \gt_{X,q}[f]\le 1
\},
    \end{equation}
is compact in $L^2({\R}_+)$. 
We divide the proof in 3 steps.

$(i_1)$ Fix  $\varepsilon>0$  and  set
\[
\N'(\varepsilon):=\{k\in\N: \, |d_k|\le \varepsilon\}\quad \text{  and}\quad
\N''(\varepsilon):=\{k\in\N: \, |d_k|> \varepsilon\}.
\]
 Clearly $\N = \N'(\varepsilon) \cup  \N''(\varepsilon)$.

Firstly, we estimate $\|f\|_{L^2(\Delta_k)}$ for   $k\in \N'(\varepsilon)$.
For any $f\in W^{1,2}(\Delta_k)$ and any $x, y \in\Delta_k$ we have
    \begin{equation}\
|f(x) - f(y)|^2 = \Big|\int^{y}_{x}f'(t)dt\Big|^2 \le |x-y|\int^{y}_{x}|f'(t)|^2 dt \le  |x-y|\cdot \|f\|^2_{W^{1,2}(\Delta_k)}.
     \end{equation}
Therefore,
   \begin{equation}\label{3.3A}
|f(x)|^2 \le  2|f(y)|^2 +  2|x-y| \cdot\|f\|^2_{W^{1,2}(\Delta_k)},\quad x, y\in\Delta_k.
   \end{equation}
Since $f$ is continuous on $\Delta_k$, there exists  $y_k\in\Delta_k$ such that
     \begin{equation}\label{3.5}
|f(y_k)|^2  = \frac{\int_{\Delta_k}q(x)|f(x)|^2dx}{\int_{\Delta_k}q(x)dx}, \qquad k\in \N'(\varepsilon).
     \end{equation}
According to the condition  \eqref{eq:4.9},  there is $p':=p'(\varepsilon)\in\N$ such that
     \begin{equation}\label{3.6}
\frac{1}{\gd_k} \int_{\Delta_k}q(x) dx  >
\frac{1}{\varepsilon}\, ,\qquad  ( k\ge p').
     \end{equation}
Setting $y=y_k$ in \eqref{3.3A} and then integrating it over $\Delta_k$,  \eqref{3.5} and \eqref{3.6} imply 
          \begin{align}\label{3.7}
\int_{\Delta_k}|f(x)|^2 dx\le&  \frac{2d_k}{\int_{\Delta_k}q(x)dx} \cdot {\int_{\Delta_k}q(x)|f(x)|^2dx}  +  2d^2_k \|f\|^2_{W^{1,2}(\Delta_k)} \nonumber  \\
&\le  2\varepsilon {\int_{\Delta_k}q(x)|f(x)|^2dx} +  2\varepsilon^2 \|f\|^2_{W^{1,2}(\Delta_k)}, \qquad k\in \N'_{p'}({\varepsilon}),
          \end{align}
where $\N'_{p'}({\varepsilon}) := \{k\in \N'({\varepsilon}): k\ge p'\}.$

$(i_2)$ Now let  $k\in\N''(\varepsilon)$. We set $n_k := \floor{d_k/\varepsilon}$, where $\floor{.}$ is the usual floor function,  and then divide the interval $\Delta_k=[x_{k-1},x_k]$ into disjoint  subintervals $\Delta_k^{(j)},$  $j\in\{1,\dots,n_k\}$,
such that
   \begin{equation}\label{3.8}
\Delta_k=\bigcup_{j=1}^{n_k}\Delta_k^{(j)},\qquad \begin{cases}
|\Delta_k^{(j)}|=\varepsilon,&   j\le n_k-1,\\
 |\Delta_k^{(n_k)}|=\varepsilon_k , & \varepsilon_k\in [\varepsilon , 2\varepsilon).
\end{cases}
   \end{equation}
Since $f$ is continuous on $\Delta_k^{(j)}$, there exists  $x_k^{(j)}\in\Delta_k^{(j)}$ such that
        \begin{equation}\label{3.9}
|f(x_k^{(j)})|^2  = \frac{\int_{\Delta_k^{(j)}}q(x)|f(x)|^2 dx}{\int_{\Delta_k^{(j)}}q(x)dx},\quad j\in\{1,\ldots,n_k\}.
        \end{equation}
Then integrating inequality \eqref{3.3A}  with $y=x_k^{(j)}$ over $\Delta_k^{(j)}$ and using  \eqref{3.8}  and  \eqref{3.9}, we obtain
     \begin{equation}\label{3.10}
\int_{\Delta_k^{(j)}}|f(x)|^2 dx  \le  2\varepsilon \frac{\int_{\Delta_k^{(j)}}q(x)|f(x)|^2 dx}{\int_{\Delta_k^{(j)}}q(x)dx}
+\  2\varepsilon^2  \|f\|^2_{W^{1,2}(\Delta_k^{(j)})}, \quad j\le n_k-1.
  \end{equation}
and
     \begin{equation}\label{3.10A}
\int_{\Delta_k^{(j)}}|f(x)|^2 dx  \le  4\varepsilon \frac{\int_{\Delta_k^{(j)}}q(x)|f(x)|^2 dx}{\int_{\Delta_k^{(j)}}q(x)dx}
+\  8\varepsilon^2  \|f\|^2_{W^{1,2}(\Delta_k^{(j)})}, \qquad j=n_k.
  \end{equation}

Now noting that $q$ satisfies \eqref{4.1}, we can find $p''=p''(\varepsilon)>0$ such that
\be\label{3.11}
\int_x^{x+\varepsilon}q(x)dx>1,\qquad x\ge p''.
\ee
Combining  \eqref{3.10}, \eqref{3.10A}  with \eqref{3.11}, we get for $k\in \N''_{p''}(\varepsilon):=\{n\in \N''(\varepsilon):\, n\ge p''\}$
      \begin{align}
\int_{\Delta_{k}}|f(x)|^2dx  &\le 4\varepsilon\sum^{n_k}_{j=1}\int_{\Delta_k^{(j)}}q(x)|f(x)|^2dx + 8\varepsilon^2 \sum^{n_k}_{j=1}\|f\|^2_{W^{1,2}(\Delta_k^{(j)})}  \nonumber \\
&= 4\varepsilon\int_{\Delta_{k}}q(x)|f(x)|^2dx + 8\varepsilon^2 \|f\|^2_{W^{1,2}(\Delta_k)} \le 8\varepsilon\|f\|^2_{W^{1,2}(\Delta_k;q)}.\label{3.12}
      \end{align}

$(i_3)$ Setting $p = p(\varepsilon) := \max \{p'(\varepsilon),p''(\varepsilon)\}$ and summing up the inequalities  \eqref{3.7}  and   \eqref{3.12}, we
arrive at the following  estimate 
     \begin{align}
\int_{x_{p-1}}^\infty |f(x)|^2dx  &\le \sum_{k=p}^\infty\int_{\Delta_{k}}|f(x)|^2dx \nonumber\\
&\le 8\varepsilon\sum^\infty_{k=p} \Big(\int_{\Delta_{k}}q(x)|f(x)|^2dx + \varepsilon \|f\|^2_{W^{1,2}(\Delta_k)} \Big)   \label{3.13}\\
&\le 8C_1 \varepsilon\|f\|^2_{W^{1,2}(\R_+\setminus X;q)} \le 8C_1 \varepsilon. \nonumber
     \end{align}
Thus the tails  $\int^{\infty}_{x_p}|f(t)|^2 dt$ are uniformly  small for $f$ belonging to the unit ball  $\mathbb{U}_{X, q}$  and hence the set
$i(U_{X,q})$ is compact in $L^2(\R_+)$. Therefore, the
embedding  $i: W^{1,2}(\R_+\setminus X; q)\hookrightarrow
L^2(\R_+)$ is compact  and, by the Rellieh theorem  (Theorem
\ref{th_Rel}), the spectrum $\sigma(\rH_{X,q}^N)$ is discrete.
  \end{proof}
        \begin{corollary}\label{cor4.3}
Let  $0<\gd_*\le \gd^*<\infty$ and  $q\in L^1_{\loc}(\R_+)$.
 Let also $q$ satisfy \eqref{eq:4.8}.  Then the operator $\rH_{X,q}^N$ has purely discrete spectrum if and only if $q$ satisfies Molchanov's condition \eqref{4.1}.
  \end{corollary}
\begin{proof}
Since $q$ satisfies \eqref{eq:4.8}, it suffices to consider the case
of a nonnegative $q$, $q=q_+$.  Let $\varepsilon\in(0,d_*)$. Then
   \[
\frac{1}{\gd_k} \int_{\Delta_k}q(x) dx=\frac{1}{\gd_k} \int_{\Delta_k}q_+(x) dx
 \ge\frac{1}{d_k}\int^{x_{k-1}+\varepsilon}_{x_{k-1}}q(x)dx \ge
\frac{1}{\gd^*} \int^{x_{k-1}+\varepsilon}_{x_{k-1}}q(x)dx
\]
for all $k\in\N$. This inequality shows that \eqref{eq:4.9} is implied by \eqref{4.1}.
It remains to apply Theorem \ref{th:discr_N}.
     \end{proof}
        \begin{corollary}\label{cor4.4}
Let $q\in L^1_{\loc}(\R_+)$ and  
satisfy condition  \eqref{eq:4.8}.  If $\gd_k \to 0$, then condition \eqref{eq:4.9}
is necessary and sufficient   for the operator $\rH_{X,q}^N$
 to have purely  discrete spectrum.
  \end{corollary}
  \begin{proof}
By Theorem \ref{th_discretcriter},  it suffices to show  that \eqref{eq:4.9} implies \eqref{4.1} if $\gd_k\to 0$.
Since $q$  satisfies condition  \eqref{eq:4.8}, we can restrict ourselves to the case
of a nonnegative $q$, $q=q_+$.

It follows from \eqref{eq:4.9} that  for any $N\in\N$ there exists $p_1 = p_1(N)$ such that
    \begin{equation}\label{3.30}
\frac{1}{\gd_k}\int_{\Delta_k} q(x)dx > N, \quad  k\ge p_1.
   \end{equation}
Fix $\varepsilon>0$. Since $d_k \to 0$, there exists $p_2= p_2(\varepsilon)$ such that
    \begin{equation}\label{3.31}
d_k\le \frac{\varepsilon}{3},\quad   k\ge p_2.
  \end{equation}
Let $p :=\max(p_1, p_2)$ and  let $x>x_p$.
%
%
%
%
Using \eqref{3.30},  \eqref{3.31} and the  non-negativity of $q$, we get 
   \begin{align*}
\int^{x+ \varepsilon}_x q(t)dt \ge &\sum_{k:\, \Delta_k\subseteq [x,x+\varepsilon]} \int_{\Delta_k} q(t)dt   \nonumber \\
&\ge
\sum_{k:\, \Delta_k\subseteq [x,x+\varepsilon]} Nd_k\ge  N\Big(\varepsilon - 2\frac{\varepsilon}{3}\Big) = N\frac{\varepsilon}{3},  \qquad  x\ge x_p.
\end{align*}
The latter implies that $q$ satisfies Molchanov's condition \eqref{4.1} since $N$ is arbitrary.
\end{proof}

The next examples show that in the case $\gd_*=0$ conditions \eqref{4.1} and \eqref{eq:4.9} complement each other, that is, neither \eqref{4.1} does imply \eqref{eq:4.9}, nor \eqref{eq:4.9} does imply \eqref{4.1} if $\gd_*=0$.
\begin{example}\label{example4.4} $(i)$  Let  $X=\{x_k\}_1^\infty$, where
\[
x_k=\begin{cases}
j, & k=2j-1,\\
j+\frac{1}{j}, & k=2j.
\end{cases}
\]
Let $q$ be given by
\[
q(x)=\begin{cases}
x, & x\in \cup_{j=1}^\infty[x_{2j-2},x_{2j-1}],\\
0, & x\in \cup_{j=1}^\infty (x_{2j-1},x_{2j}).
\end{cases}
\]
Clearly, $\int_{x_{2j-1}}^{x_{2j}}q(x)dx=0$ for all $j\in \N$ and hence $q$ does not satisfy \eqref{eq:4.9}. However, for any $\varepsilon>0$
\[
\lim_{x\to\infty}\frac{1}{x}\int_x^{x+\varepsilon} q(x)dx =\varepsilon,
\]
which yields  Molchanov's condition  \eqref{4.1}.

\item $(ii)$  On the other hand, let $X=\N$ and let $q$ be given by
\[
q(x)=\begin{cases}
2x, & x\in \cup_{n\in\N}[n-1,n-\frac12),\\
0, & x\in \cup_{n\in\N}[n-\frac12,n)\,.
\end{cases}
\]
Clearly, $\gd_k=k-(k-1)=1$, $k\in\N$,  and
\[
\frac{1}{\gd_k}\int_{\Delta_k}q(x)dx=\int_{k-1}^{k-\frac12}2xdx=k-\frac34\to +\infty, \quad k\to \infty\,.
\]
Therefore, $q$ satisfies \eqref{eq:4.9}. However, $q$ does not satisfy \eqref{4.1} whenever $\varepsilon<\frac{1}{2}$.
           \end{example}

\subsection{Sufficient conditions}\label{ss:3.3}
Now we are in position to prove the sufficient part of Theorem \ref{th_discretcriter}.

\begin{proof}[Proof of Theorem \ref{th_discretcriter}. Sufficiency.]
 If $q$ and $\gB$ satisfy \eqref{eq:b_klmn}, then, by Proposition \ref{cor:3.3}(ii), $\dom(\gt_{X,\gB,q})=\dom(\gt_{X,\gB_+,q_+})$ algebraically and topologically.  Moreover, by Lemma \ref{lem_III.3},
 $\dom(\gt_{X,\gB_+,q_+})$ (and hence $\dom(\gt_{X,\gB,q})$) is continuously embedded into $W^{1,2}(\R_+\setminus X;q_+)$. By Theorem \ref{th:discr_N}, the spectrum of $\rH_{X,q_+}^N$ is discrete and hence, by Theorem \ref{th_Rel}, $W^{1,2}(\R_+\setminus X,q)$ is compactly embedded into $L^2(\R_+)$. Therefore, we conclude that $\gH_{X,\gB,q}=\dom(\gt_{X,\gB,q})$ is compactly embedded into $L^2(\R_+)$. Theorem \ref{th_Rel} completes the proof.
 \end{proof}


     \begin{corollary}\label{cor4.1}
Assume that $q\in L^1_{\loc}(\R_+)$ and
   \begin{equation}\label{4.15}
q(x)\to\infty \quad  \text{as} \quad  x\to+\infty.
  \end{equation}
Assume also that $d^*<\infty$ and $\gB$ satisfies \eqref{I_brinck}.
Then the Hamiltonian  $H_{X,\beta,q}$ is self-adjoint and its spectrum is purely discrete.
   \end{corollary}
   \begin{proof}
Clearly condition \eqref{4.15} yields both condition \eqref{4.1} and \eqref{eq:4.9}.
Moreover, since $\gB$ satisfies \eqref{eq:4.8}, we conclude that the operator $\rH_{X,\gB,q}$ is lower semibounded. Theorem \ref{th_discretcriter} completes the proof.
  \end{proof}

  \begin{corollary}\label{cor3.9}
Assume that $0<\gd_*\le \gd^*<+\infty$ and $\inf_k\big(1/\gB_k\big)^-<\infty$. Assume also that $q\in L^1_{\loc}(\R_+)$ and $q$ satisfies \eqref{eq:4.8}. Then the operator $\rH_{X,\gB,q}$  has purely discrete spectrum if and only if $q$ satisfies Molchanov's condition \eqref{4.1}.
  \end{corollary}

  \begin{proof}
 Since $0<\gd_*\le \gd^*<+\infty$ and $q_-$ satisfies \eqref{eq:4.8}, by Corollary \ref{cor4.3}, conditions \eqref{4.1}and \eqref{eq:4.9} coincide. Therefore, Molchanov's condition becomes necessary and sufficient for the discreteness.
  \end{proof}

   \begin{remark}\label{rem4.7}
 In the case  $q\in L^\infty(\R_+)$, Hamiltonians
 $\rH_{X,\gB,q}$ with discrete spectrum have been investigated
 in the recent publications \cite{KosMal09}, \cite{KosMal10}.
  It is shown in \cite[Theorem 6.8]{KosMal09}
that $\sigma(\rH_{X,\gB,q})$ is discrete if and only if $d_n\to 0$ and the
spectrum of a certain Jacobi matrix is discrete. Moreover, the corresponding matrix can be considered as a Krein--Stieltjes string operator. Hence, using the Kac--Krein discreteness criterion \cite{KK71}, several necessary and sufficient discreteness conditions have been obtained in \cite{KosMal09}.
Although Theorem  \ref{th_discretcriter} gives an affirmative answer
for semibounded Hamiltonians $\rH_{X,\gB,q}$, it does not cover the results
from \cite[Section 6.4]{KosMal09}, since the Hamiltonians in \cite{KosMal09} are not
assumed to be lower semibounded. For instance, it was shown in \cite[Propositions 6.9]{KosMal09} that 
the spectrum of $\rH_{X,\gB,q}$ is not discrete if there is $C>0$ such that
\[
\gB_k^-\ge C\left(\frac{1}{\gd_k}+\frac{1}{\gd_{k+1}}\right),\qquad x_k\in X_-.
\]
On the other hand,  if
\[
\gB_k\ge -\gd_k, \quad k\in\N,
\]
then, by \cite[Propositions 6.11]{KosMal09}, the spectrum of $\rH_{X,\gB,q}$ is discrete precisely if
\[
\lim_{k\to \infty}x_k\sum_{i=k}^\infty \gd_j^3=0 \quad \text{and}\quad \lim_{k\to \infty}x_k\sum_{i=k}^{\infty}(\gB_i+\gd_i)=0.
\]
\end{remark}

%
%
\section{Essential spectrum }\label{sec:cs}
\subsection{Two lemmas}\label{ss:5.1}
In this section we shall present and prove two preliminary lemmas, which may also be of independent interest.

    \begin{lemma}\label{lem5.1}
Let $d^*<\infty$ and the potential $q$ satisfy \eqref{eq:4.8}, that is
\begin{equation}\label{5.8}
\sup_{k\in \N}\frac{1}{d_k}\int_{\Delta_k}q_-(x)dx  < \infty.
    \end{equation}

 Then the mapping $i_{X,\beta}:\   W^{1,2}(\R_+\setminus X;q) \to \ell^2(\N;
|\beta|^{-1})$ given by
    \begin{equation}\label{5.1}
 (i_{X,\beta}f)_k :=
f(x_k+)-f(x_k-),\quad k\in\N,
  \end{equation}
is compact provided  that
    \begin{equation}\label{5.2}
\lim_{k\to\infty} \frac{|\beta_k|^{-1}}{\min\{d_k, d_{k+1}\}} =0.
        \end{equation}

If  additionally
  \begin{equation}\label{5.2A}
\sup_{k\in \N}\frac{1}{d_k}\int_{\Delta_k}|q(x)|dx =: C_0 < \infty,
   \end{equation}
then condition \eqref{5.2} is also necessary for the mapping
$i_{X,\beta}$ to be compact.
          \end{lemma}
   \begin{proof}
   Since $q$ satisfies \eqref{5.8},  Corollary \ref{cor:3.5} implies that
%
%
    \begin{equation}\label{5.13}
\gH_{X,q} = W^{1,2}(\R_+\setminus X; q) = W^{1,2}(\R_+\setminus X;
q_+) = \gH_{X, q_+},
        \end{equation}
which holds algebraically  and topologically.
Therefore, the mappings $i_{X,\beta}$ and
   \begin{equation}
 i_{X,\beta}^+: W^{1,2}(\R_+\setminus X;q_+)  \to \ell^2(\N; |\beta^{-1}|),
   \end{equation}
defined by the same formula \eqref{5.1}, are compact simultaneously. Moreover, due to \eqref{5.8}, $q_+$ and $q$
satisfy \eqref{5.2A} only simultaneously. Therefore, it suffices to prove Lemma \ref{lem5.1} for nonnegative potentials.

$(i)$ \emph{Sufficiency.} Without loss of generality we can assume
that $q\ge 1$. Let $\gt_{X,q}$ be the form given by \eqref{6.5B}--\eqref{6.5C}. 
If \eqref{5.2} is satisfied, then by Lemma \ref{lem_III.1}, the
mapping   $i_{X,\beta}$ is bounded. Consider the unit ball
in $W^{1,2}(\R_+\setminus X; q)$,
    \begin{equation}
\mathbb{U}_{X, q} :=  \left\{f\in  W^{1,2}(\R_+\setminus X;q): \ \gt_{X,q}[f]\le 1
\right\}.
    \end{equation}

Let us show that $i_{X, \beta}\bigl(\mathbb{U}_{X,q}\bigr)$ is
compact in $\ell^2(\N; |\beta|^{-1})$.
Notice that for any  $N\in \N$  the restriction of $i_{X,\gB}$ onto $\oplus^N_1 W^{1,2}[x_{k-1},x_k]$
%
%
%
is a bounded finite rank operator.
Hence $i_{X,\gB}$ is compact precisely if its restriction  to
$\oplus^\infty_{N+1} W^{1,2}[x_{k-1},x_k]$,  i.e., the mapping
    \begin{equation*}
i_{X, \beta}:\  \bigoplus^\infty_{N+1} W^{1,2}[x_{k-1},x_k] \to   \ell^2(\N; |\beta|^{-1})  
    \end{equation*}
is also compact. Therefore, it suffices to show that the tails
    \begin{equation}\label{5.4pr}
\sum_{k= N}^\infty \frac{ |f(x_k+)-f(x_k-)|^2}{|\beta_k|}, \qquad f\in
\mathbb{U}_{X,q},
    \end{equation}
of   $i_{X,\beta}(f)$  tend to zero 
uniformly in $f\in \mathbb{U}_{X,q}$.

  Applying inequality
\eqref{3.3} with $\varepsilon=\frac{1}{2}$, we get
   \begin{align}
&\frac{|f(x_k +) - f(x_k -)|^2}{|\beta_k|}  \le
2\frac{|f(x_k +)|^2 +| f(x_k -)|^2}{|\beta_k|} \nonumber \\
&\quad \le \frac{\gd_{k}\|f'\|^2_{L^2(\Delta_{k})} + C\gd^{-1}_{k}\|f\|^2_{L^2(\Delta_{k})}}{|\beta_k|}
+ \frac{\gd_{k+1}\|f'\|^2_{L^2(\Delta_{k+1})} + C\gd^{-1}_{k+1}\|f\|^2_{L^2(\Delta_{k+1})}}{|\beta_k|}\nonumber\\
&\qquad \quad \le \frac{{d^*}^2(\|f'\|^2_{L^2(\Delta_{k})}+\|f'\|^2_{L^2(\Delta_{k+1})}) + C(\|f\|^2_{L^2(\Delta_k)}+\|f\|^2_{L^2(\Delta_{k+1})})}{|\beta_k| \min\{\gd_k,\gd_{k+1}\}} \nonumber\\
&\qquad \quad \le C_1\frac{\|f\|^2_{W^{1,2}(\Delta_{k})} +
\|f\|^2_{W^{1,2}(\Delta_{k+1})}}{|\beta_k| \min\{\gd_k,\gd_{k+1}\}}.\label{5.3}
\end{align}
Here $C_1 := \max\{(d^*)^{2},  C\}$.
By \eqref{5.2}, for each $\varepsilon>0$ there exists $N=
N(\varepsilon) \in\N$ such that
   \begin{equation}\label{5.4}
\frac{1}{|\beta_k| \min\{\gd_k,\gd_{k+1}\}}<\varepsilon,\quad
k\ge N.
   \end{equation}
Combining \eqref{5.3} with \eqref{5.4}, we arrive at the desired estimate
    \begin{align*}
\sum^{\infty}_{k=N}\frac{|f(x_k +)-f(x_k -)|^2}{|\beta_k|} \le&
2\varepsilon C_1\sum^{\infty}_{k=N}\|f\|^2_{W^
{1,2}(\Delta_{k})} 
\\
\le& 2\varepsilon
C_1\|f\|^2_{W^{1,2}(\R_+\setminus X)}\le 2\varepsilon C_1.
      \end{align*}
%
%

\emph{Necessity}. Firstly, observe that condition \eqref{5.2} is equivalent to the following one
\begin{equation}\label{eq:5.2B}
\lim_{k\to\infty}\frac{\gd_k^{-1}}{\min\{|\gB_{k-1}|,|\gB_{k}|\}}=0.
\end{equation}
Assume the converse, i.e., conditions \eqref{5.2} and hence \eqref{eq:5.2B}
are violated. Then there exists $\varepsilon_0>0$ and a subsequence
$\{k_j\}^{\infty}_{j=1}$ such that
    \begin{equation}\label{5.5}
\frac{\gd_{k_j}^{-1}}{\min\{|\gB_{k_j-1}|,|\gB_{k_j}|\}}\ge \varepsilon_0,
\qquad j\in\N.
    \end{equation}

Again, without loss of generality we can assume that $q\ge 1$ on $\R_+$. Consider the family of indicator functions $\{h_{k}\}^{\infty}_{k=1}$ defined by
\eqref{4.9A}. Noting that $q$ satisfies \eqref{5.2A}, we get 
    \begin{equation}\label{5.6}
\|h_{k_j}\|^2_{W^{1,2}(\R_+\setminus X;q)} = \|h_{k_j}\|^2_{W^{1,2}(\Delta_k;q)}=
\frac{1}{d_{k_j}}\int_{\Delta_{k_j}}q(x)dx \le C_0,\quad
j\in\N.
   \end{equation}
On the other hand, taking \eqref{5.5} into account we get
    \begin{align}
\|i_{X,\beta}(h_{k_j})\|^2_{\ell^2(|\beta^{-1}|)} =
\frac{1}{\gd_{k_j}}&\sum^{\infty}_{m=1}\frac{|{h}_{k_j}(x_m+) - {h}_{k_j}(x_m -)|^2}{|\beta_m|}\nonumber \\
&=\frac{1}{|\beta_{k_j-1}| \gd_{k_j}} + \frac{1}{|\beta_{k_j }|
d_{k_j}} > \varepsilon, \qquad j\in\N.\label{5.7}
    \end{align}

Clearly, the system $\{h_{k_j}\}^{\infty}_{j=1}$ is orthogonal
in $W^{1,2}(\R_+\setminus X;q)$ and, by \eqref{5.6}, it is also
bounded.   Therefore it converges weakly to zero in $W^{1,2}(\R_+\setminus X; q)$. Hence, if the mapping $i_{X,\beta}$ is
compact, one has
$\|i_{X,\beta}(h_{k_j})\|_{\ell^2(|\beta^{-1}|)}\to 0$  as
$j\to\infty$. The latter  contradicts \eqref{5.7}.
      \end{proof}
%
%
%
%
\begin{lemma}\label{lem5.2}
Let $q\in L^1_{\loc}(\R_+)$  and  $d^*<\infty$.  Then the
embedding
    \begin{equation}
i_{X, q}:\   W^{1,2}(\R_+\setminus X) \to  L^2(\R_+;|q|),
\qquad  i_{X,q}(f)=f,
   \end{equation}
is compact if and only if
   \begin{equation}\label{5.20}
\frac{1}{d_k}\int_{\Delta_k}|q(x)|dx \to 0  \qquad \text{as}\qquad
k\to\infty.
   \end{equation}
\end{lemma}
   \begin{proof}
\emph{Sufficiency.}  
To prove the compactness of the mapping $i_{X,q}$
it suffices to show that the unit ball $\mathbb{U}_{X}$ in $W^{1,2}(\R_+\setminus X)$
is compact in $L^2(\R_+;|q|)$.   In turn, it suffices to show
that the tails $\int^{\infty}_N|f|^2 dx$ are uniformly small in
$f\in \mathbb{U}_{X}$.

Setting $\varepsilon=1$ in \eqref{3.3}, we obtain
   \begin{equation}
|f(x)|^2 \le d_k\|f'\|^2_{L^2(\Delta_k)} +  Cd^{-1}_k\|f\|^2_{L^2(\Delta_k)},\quad f\in W^{1,2}(\Delta_k), \quad x\in\Delta_k.
\end{equation}
Multiplying  this inequality by $|q(x)|$ and then integrating over $\Delta_k$, we get
  \begin{align}
\int_{\Delta_{k}}|q(x)||f(x)|^2 dx 
\le &\Big(\frac{1}{d_k}\int_{\Delta_k}|q(x)|dx\Big)\Big( d^2_k \|f'\|^2_{L^2(\Delta_k)}
     +\  C  \|f\|^2_{L^2(\Delta_k)} \Big)\nonumber   \\
     &\le C_1 \frac{1}{d_k}\int_{\Delta_k}|q(x)|dx\cdot \|f\|^2_{W^{1,2}(\Delta_k)}, \quad k\in \N,\label{eq:5.18}
 \end{align}
 where $C_1 := \max\{C, {d^*}^2\}$.

Further, by \eqref{5.20}, for each $\varepsilon>0$
there exists $N=N(\varepsilon)\in\N$ such that
    \begin{equation}\label{5.23}
\frac{1}{d_k}\int_{\Delta_k}|q(x)|dx  <\varepsilon \qquad
\text{for}\qquad  k\ge N.
  \end{equation}
Combining \eqref{eq:5.18} with \eqref{5.23}, we obtain for $f\in \mathbb{U}_{X}$
      \begin{align}
\int^{\infty}_{x_{N}}|f(x)|^2|q(x)|dx &= \sum^{\infty}_{k=N+1}\Big(\frac{1}{d_k}\int_{\Delta_k}|q(x)|dx \cdot \|f\|^2_{W^{1,2}(\Delta_k)} \Big)\nonumber   \\
&\le \varepsilon \sum^{\infty}_{k=N+1} \|f\|^2_{W^{1,2}({\Delta_k})} \le \varepsilon\|f\|^2_{W^{1,2}(\R_+\setminus X)}  \le \varepsilon.
      \end{align}
%
%
%
%
This proves compactness of the set $i(\mathbb{U}_{X})$ in $L^2(\R_+;|q|)$.

          \emph{Necessity}.
Assume that the embedding $i_{X, q}:  W^{1,2}(\R_+\setminus X) \to
L^2(\R_+;|q|)$ is compact.   Assume also that
condition \eqref{5.20} is violated. Then there exists
$\varepsilon_0>0$ and a subsequence $\{k_j\}^{\infty}_{j=1}$ such
that
    \begin{equation}
\frac{1}{\gd_{k_j}}\int_{\Delta_{k_j}}|q(x)|dx \ge \varepsilon_0, \qquad
j\in\N.
    \end{equation}
Consider the family $\{h_{k_j}\}^{\infty}_{j=1}$ given by
\eqref{4.9A}.   Clearly,
     \be\label{5.25}
\|i_{X,q}({h}_{k_j})\|^2_{L^2(\R_+;|q|)} =
\|{h}_{k_j}\|^2_{L^2(\R_+;|q|)} =
\frac{1}{d_{k_j}}\int_{\Delta_{k_j}}|q(x)|dx \ge \varepsilon_0,
\quad j\in\N.
     \ee
Since the system  $\{h_{k_j}\}^{\infty}_{j=1}$  is orthonormal
in  $W^{1,2}(\R_+\setminus X)$,  it weakly converges to
zero in $W^{1,2}(\R_+\setminus X)$. Since the embedding
$i_{X,q}$ is compact, $\|i(h_{k_j})\|_{L^2(\R_+;|q|)}
=\|{h}_{k_j}\|_{L^2(\R_+;|q|)}\to 0$ as $j\to \infty$.
This contradiction completes the proof.
       \end{proof}
    \begin{remark}
According to Lemma \ref {lem5.2},  the embedding  $W^{1,2}(\R_+\setminus X) \hookrightarrow L^2(\R_+;|q|)$
is compact whenever \eqref{5.20} holds. Since the embedding $W^{1,2}(\R_+) \hookrightarrow W^{1,2}(\R_+\setminus X)$ is continuous, condition \eqref{5.20}  yields compactness of the embedding $W^{1,2}(\R_+) \hookrightarrow L^2(\R_+;|q|)$.
On the other hand,
according to the Birman  result \cite{Bir61} (see also \cite[Theorem II.2]{Gla65}), the embedding
$W^{1,2}(\R_+) \hookrightarrow L^2(\R_+;|q|)$  is compact  if and only if
     \begin{equation}\label{5.24}
\lim_{x\to\infty}\int^{x+1}_x|q(t)|dt=0.
     \end{equation}
Thus,  condition \eqref{5.20} implies  condition \eqref{5.24}.

Let us show  this fact by a direct proof.
For any $\varepsilon >0$ choose $N=N(\varepsilon)\in\N$ as in \eqref{5.23}. Then
%
%
    \begin{align*}
\int^{x+1}_x|q(t)|dt
 &\le\sum_{k:\, \Delta_k\cap[x,x+1]\neq\emptyset} \int_{\Delta_k}|q(t)|dt \\
 \le&\sum_{k:\, \Delta_k\cap[x,x+1]\neq\emptyset}\varepsilon d_k\le\varepsilon\bigl(1+2
d^*\bigr), \quad x\ge x_N,
    \end{align*}
which implies \eqref{5.24}.
   \end{remark}

\subsection{Essential spectrum of $\rH_{X,q}^N$}\label{ss:4.2}
As we already saw in the previous sections, spectral properties of the Hamiltonian $\rH_{X,\gB,q}$ are
closely related with those of the Neumann realization $\rH_{X,q}^N$.
In particular, their essential spectra are closely related too. In this section we collect
some results describing the spectrum of $\rH_{X,q}^N$.

We begin with the following simple fact.
  \begin{lemma}\label{lem:5.4}
Let the operator $\rH_{X,q}^N$ be given by \eqref{eq:H^N}--\eqref{eq:HkN2}.
Then $\sigma(\rH_{X,q}^N)$ is pure point and 
\be
\sigma(\rH_{X,q}^N)=\overline{\Big(\bigcup_{k=1}^\infty\sigma(\rH_{k,q}^N)\Big)},\quad
\sigma_{\ess}(\rH_{X,q}^N)=\sigma_{\ess}^\infty(\rH_{X,q}^N)\bigcup\Big(\bigcup_{k=1}^\infty\sigma(\rH_{k,q}^N)\Big)',
\ee
where $\sigma_{\ess}^\infty(\rH_{X,q}^N)$ is the set of eigenvalues having infinite multiplicity, i.e.,
\be
\sigma_{\ess}^\infty(\rH_{X,q}^N):=\{\lambda\in \sigma(\rH_{X,q}^N):\ \dim(\ker(\rH_{X,q}^N-\lambda I))=\infty\}.
\ee
  \end{lemma}
\begin{proof}
Immediately follows from the definition \eqref{eq:H^N} of $\rH_{X,q}^N$ and the fact that for each $k\in \N$ the spectrum of $\rH_{X,q}^N$ is simple and discrete.
  \end{proof}
For a sequence $\{\gd_k\}_{k=1}^\infty$ we define the following set
\be\label{eq:d1}
\mathcal{D}:=\mathcal{D}^\infty\cup \big(\mathcal{D'}\setminus\{0\}\big),
\ee
where 
\be\label{eq:d2}
\mathcal{D}^\infty=\{d\in\R_+:\ d\ \text{appears infinitely  many times in}\ \{\gd_k\} \}.
  \ee
%

Introducing the multiplication operator  $\mathcal{M}_d:\ell^2(\N)\to \ell^2(\N)$,\
$(\mathcal{M}_df)_k=\gd_kf_k$, $k\in\N$,  we can rewrite the  definition of set $\mathcal{D}$  as follows
%
%
\be
\mathcal{D}=\sigma_{\ess}(\mathcal{M}_d)\setminus\{0\}.
\ee
Emphasize that $\mathcal{D}$ is a (not necessarily countable) subset of $[\gd_*, \gd^*]$.

%
%
%
    \begin{proposition}\label{th5.2}
Let $q\in L^1_{\loc}(\R_+)$ and $\gd^*<\infty$. If $q$ satisfies
condition \eqref{5.20}, then
\be
\sigma_{\ess}(\rH_{X,q}^N) = \sigma_{\ess}(\rH_{X}^N)=\{0\}\bigcup_{\lambda\in \mathcal{D}}\Big\{\frac{\pi^2n^2}{\lambda^2}\Big\}_{n=1}^\infty.
\ee
  \end{proposition}
   \begin{proof}
   Recall that the form $\gt_X:=\gt_{X,0}$ associated with
the Hamiltonian  $\rH_X^N  = \bigoplus^{\infty}_{k=1}\rH^N_k$ is given
by \eqref{eq:H^N} with $q=\bold{0}$. Consider quadratic forms $\gq$ and
$|\gq|$ given by \eqref{q_form} with $q$ and $|q|$,  respectively.
%
%
%
%
%
%
%
%
By Lemma \ref{lem5.2},  the form $|\gq|$ is $\gt_X$-compact
since $\dom(\gt_X) =  W^{1,2}(\R_+\setminus X)$.

Further, note that  $|\gq[f]|
\le |\gq|[f]$ for each $f\in \dom(|\gq|)$ and hence $\dom(|\gq|)\subseteq \dom(\gq)$. Therefore, by Lemma
\ref{lem3.2}, the form $\gq$ is also $\frak t_X$-compact. Since
the operator $\rH_{X,q}^N$ is associated with the form $\gt_X +
\gq$,  Theorem \ref{th2.2} implies $\sigma_{\ess}(\rH_{X,q}^N) =
\sigma_{\ess}(\rH_{X}^N)$. To complete the proof it suffices to mention that the spectrum of $\rH_k^N$ is simple and is given by  $\sigma(\rH^N_k)= \{(\frac{\pi
n}{d_k})^2\}^{\infty}_{n=0}$.
   \end{proof}
\begin{corollary}\label{cor:5.6}
Let $q\in L^1_{\loc}(\R_+)$ and $\gd^*<\infty$. Assume also that $q$ satisfies
condition \eqref{5.20}. Then
\be\label{eq:5.26}
\sigma_{\ess}(\rH_{X,q}^N) = \{0\}
\ee
if and only if  $\lim_{k\to\infty} \gd_k =0$.
\end{corollary}
\begin{proof}
By Proposition \ref{th5.2}, \eqref{eq:5.26} holds precisely if (see \eqref{eq:d1} and \eqref{eq:d2})
    \[
\mathcal{D}=\emptyset.
    \]
Since $\gd^*<\infty$, the latter holds precisely when $\gd_{k}\to 0$.
%
%
%
\end{proof}
Consider now the periodic case.
\begin{corollary}\label{cor:5.7}
Let $X=a\N$ and let $q\in L^1_{\loc}(\R_+)$ be an $a$-periodic function, i.e., $q(x+a)=q(x)$ for a.a. $x\in \R_+$. Then
\be\label{4.31}
\sigma(\rH_{X,q}^N)=\sigma_{\ess}(\rH_{X,q}^N) = \sigma_{\ess}^{\infty}(\rH_{X,q}^N) = \sigma(\rH_{q,1}^N).
\ee
In particular, if $q=\bold{ c}$, then
\[
\sigma(\rH_{X,\bold{c}}^N)=\sigma_{\ess}(\rH_{X,\bold{c}}^N)=\sigma_{\ess}^{\infty}(\rH_{X,\bold{c}}^N)=\Big\{\frac{\pi^2n^2}{a^2}+c\Big\}_{n=0}^\infty.
\]
  \end{corollary}
\begin{proof}
Since $q$ is $a$-periodic,  the operators $\rH_{q,k}^N$ and $\rH_{q,1}^N$
are unitarily equivalent and  $\sigma(\rH_{q,k}^N) = \sigma(\rH_{q,1}^N)$, $k\in \N$.
This proves  \eqref{4.31}. If $q=\bold{ c}$, then  $\sigma(\rH_{q,1}^N) = \{\frac{\pi^2n^2}{a^2}+c\}_{n=0}^\infty$.
\end{proof}

\subsection{Proof of Theorem \ref{thContSpec}}\label{ss:4.3}
Now we are ready to prove Theorem \ref{thContSpec}.
%
          \begin{proof}[Proof of Theorem \ref{thContSpec}]
Condition \eqref{5.20} implies that $q_-$ also satisfies \eqref{5.20} and hence $q$ satisfies \eqref{5.8},
i.e., the first condition in \eqref{eq:b_klmn}. Therefore, by
Corollary \ref{cor:3.5}, the operator $\rH_{X,q}^N$ is lower semibounded and the corresponding  energy space $\gH_{X,q}$ coincides with  $W^{1,2}(\R_+\setminus X;q)$ algebraically and topologically (see
\eqref{5.13}).

Next notice that, by Lemma \ref{lem5.1}, the form $|\gb_{X}|$,
   \begin{equation}
|\gb_X|[f]:=\gb_X^+[f]+\gb_X^-[f]=\sum^{\infty}_{k=1}\frac{|f(x_k +)-f(x_k-)|^2}{|\beta_k|},
   \end{equation}
is $\gt_{X,q}$-compact, i.e., it is compact on $\gH_{X,q} =
W^{1,2}(\R_+\setminus X;q)$.
Since
\[
|\gb_X[f]|\le |\gb_X|[f],\qquad f\in \gH_{X,q}=
W^{1,2}(\R_+\setminus X;q),
\]
by Lemma \ref{lem3.2}, the form $\gb_X$ is $\gt_{X,q}$-compact too.
Next, applying Theorem \ref{th2.2}, we arrive at the
equality $\sigma_{\ess}(\rH_{X,\beta, q}) =
\sigma_{\ess}(\rH_{X,q}^N)$. On the other hand, by Proposition
\ref{th5.2},   condition \eqref{5.20} yields
$\sigma_{\ess}(\rH_{X,q}^N) = \sigma_{\ess}(\rH_{X}^N)$. Combining
both relations we arrive at   \eqref{1.12}.
%
           \end{proof}
Consider several simple examples. We begin with a Kronig--Penney type model.
\begin{corollary}\label{cor:5.8}
Let $X=a\N$ and let $q\in L^1_{\loc}(\R_+)$ be an $a$-periodic function, i.e., $q(x+a)=q(x)$ for a.a. $x\in \R_+$. If
\begin{equation}\label{5.26}
\lim_{k\to\infty}|\beta_k|=\infty,
    \end{equation}
then the spectrum of the operator $\rH_{X,\gB,q}$ is pure point and
\be
\sigma_{\ess}(\rH_{X,\gB,q})=\sigma_{\ess}(\rH_{X,q}^N) = \sigma_{\ess}^{\infty}(\rH_{X,q}^N) = \sigma(\rH_{q,1}^N).
\ee
In particular, if $q=\bold{ c}$, then
\[
\sigma_{\ess}(\rH_{X,\gB,\bold{c}})=\Big\{\frac{\pi^2n^2}{a^2}+c\Big\}_{n=0}^\infty.
\]
\end{corollary}

\begin{proof}
Since $\gd_k=a$ for all $k\in \N$, condition  \eqref{5.26} implies \eqref{5.2}. Moreover, $q$ satisfies \eqref{5.8} since it is $a$-periodic.
One completes the proof by combining  Theorem \ref{thContSpec} with Corollary \ref{cor:5.7}.
  \end{proof}
The following  result is immediate from Theorem \ref{thContSpec}.
        \begin{corollary}\label{cor5.1}
Let $0<\gd_*\le \gd_*<\infty$ and $q$ satisfy  \eqref{5.8}.
 Then:
\begin{itemize}
\item[(i)]
$\sigma_{\ess}(\rH_{X,\beta,q}) =
\sigma_{\ess}(\rH_{X,q}^N)$ whenever \eqref{5.26} holds.
\item[(ii)]    
    If, in addition, $q$ satisfies \eqref{5.20},  then
    \[
    \sigma_{\ess}(\rH_{X,\beta,q}) =\bigcup_{\lambda\in\mathcal{D}}\Big\{\frac{\pi^2n^2}{\lambda^2}\Big\}_{n=0}^\infty,
    \]
    where $\mathcal{D}$ is given by \eqref{eq:d1} and \eqref{eq:d2}.
    \end{itemize}
       \end{corollary}
       \begin{proof}
Since  $d_*>0$,  condition \eqref{5.26} yields \eqref{5.2}.
 Theorem \ref{thContSpec} completes the proof.
       \end{proof}
  In conclusion let us present a class of Hamiltonians $\rH_{X,\gB,q}$ with pure point spectra and such that their eigenvalues accumulate only at $0$ and $\infty$.

  \begin{corollary}\label{cor:5.10}
  Let  $X$ be such that
  \[
  \lim_{k\to\infty}\gd_k=0. 
  \]
  Let also $q $ satisfy \eqref{5.20}.
  If $\gB=\{\gB_k\}_1^\infty$ satisfies \eqref{5.2}, then the spectrum of
  $\rH_{X,\gB,q}$ is pure point and accumulates only at $0$ and $\infty$, that is
  \be
\sigma_{\ess}(\rH_{X,\gB,q}) = \sigma_{\ess}(\rH_{X,q}^N) = \{0\}.
  \ee
  \end{corollary}
  \begin{proof}
  Follows by combining Theorem \ref{thContSpec}  with Corollary \ref{cor:5.6}.
  \end{proof}

  \subsection{Negative spectrum and $h$-stability}\label{ss:4.4}
In this section we investigate the negative spectrum of the Hamiltonian
$\rH_{X,\gB,q}$.
Let us note that the negative spectrum of the
operator $\rH_{X,\gB}=\rH_{X,\gB,\bold{0}}$ has been studied in \cite{AlbNiz03b}
(the case of a finite number of point interactions) and
in \cite{GolOri10}, \cite{KosMal10}, \cite{AlbNiz06}, \cite{BraNiz11} (the case of an infinitely many point interactions).

We begin with the generalization of Birman's result on
the discreteness of  negative spectrum (cf. \cite{Bir61} and
\cite[Theorem 2.3]{Gla65}).

     \begin{proposition}\label{prop6.2}
Let $\gd^*<\infty$ and $q\in L^1_{\loc}(\R_+)$. If the negative parts $q_-$ and $\gB^-$ of $q$ and $\gB$ satisfy \eqref{5.20} and \eqref{5.2}, respectively, that is
     \begin{equation}\label{6.3}
\lim_{k\to\infty}\frac{1}{\gd_k}\int^{x_{k}}_{x_{k-1}} q_-(t)dt=0 \qquad 
     \end{equation}
and
     \begin{equation}\label{4.40} 
 \lim_{k\to\infty}\frac{(\gB_k^{-1})^-}{\min\{\gd_k,\gd_{k+1}\}}= 0,
     \end{equation}
%
%
then the negative part of the spectrum $\sigma(\rH_{X,\gB,q})$ is
bounded from below and discrete.
\end{proposition}
   \begin{proof}
Firstly, by Proposition \ref{cor:3.3}, condition  \eqref{6.3} implies lower semiboundedness of  the operator $\rH_{X,\gB,q}$.
 Next, since $q_-$ satisfies \eqref{6.3} and using Lemma \ref{lem5.2} and Theorem \ref{th2.2}, we conclude
that the negative spectrum of the operator $\rH_{X,q}^N$ is bounded from below and discrete.
Note that the latter is equivalent
to the inclusion $\sigma_{\ess}(\rH_{X,q}^N)\subseteq{\R}_+$.
In particular, this immediately implies that $\sigma_{\ess}(\rH_{X,\gB^+,q})\subseteq{\R}_+$.

By Lemma \ref{lem5.1},  the second condition in \eqref{6.3} yields
compactness of the operator $i_{X, \gB^-}:\  W^{1,2}(\R_+\setminus X; q)
\to \ell^2((\gB^{-1})^-)$ defined by \eqref{5.1} with $(\gB^{-1})^-$ in place of $\gB^{-1}$.
{Since the   
embedding}  $i_2: {\gH}_{X,\gB^+,q} \to {\gH}_{\rH_{X,q}^N} = W^{1,2}(\R_+\setminus X;q)$ is continuous,
the restriction
$$
i_{X, \gB^-}\upharpoonright
{\gH}_{X,\gB^+,q} = i_{X, \gB^-}i_2:\ {\gH}_{X,\gB^+,q} \to \ell^2((\gB^{-1})^-)
$$
is compact too.
By Lemma \ref{lem5.1}, the form $\gb^-_{X}$ is compact on ${\gH}_{X,\gB^+,q}$. Therefore, by Theorem
\ref{th2.2}, $\sigma_{\ess}(\rH_{X,\gB,q})=\sigma_{\ess}(\rH_{X,\gB^+,q})\subseteq{\R}_+$ and the negative spectrum of $\rH_{X,\gB,q}$ is discrete.
   \end{proof}
    \begin{corollary}
Assume that $0<\gd_*\le \gd^*<\infty$ and  $q$ satisfies \eqref{6.3}. If
%
%
\be\label{4.37}
\lim_{k\to\infty}(\gB_k^{-1})^- = 0,
\ee
then the negative part of the spectrum $\sigma(\rH_{X,\gB,q})$ of
the Hamiltonian $\rH_{X,\gB,q}$  is bounded from below and
discrete.
        \end{corollary}
\begin{proof}
Note that \eqref{4.37} implies \eqref{4.40} since  $\gd_*>0$.
Proposition \ref{prop6.2} completes the proof.
       \end{proof}
%
%
%
%

Our next aim is to show that condition \eqref{5.2} is in a sense necessary
for the validity of the equality  
$\sigma_{\ess}(\rH_{X,\gB,q})=\sigma_{\ess}(\rH_{X,q}^N)$.

Let $h\in\R_+$. We are going to investigate the property of stability of $\sigma_{\ess}(\rH_{X,q}^N)$ of the Hamiltonian $\rH_{X,q}^N$ under
perturbation by the form $h\gb_{X}:=h(\gb^+_{X}+\gb^-_{X})$ defined by \eqref{6.5},
where $h\gB:=\{h\gB_k\}_1^\infty$
Observe that $h\gb_{X}=(h\gb)_{X}$. Thus, we investigate
the essential  spectra $\sigma_{\ess}(\rH_{X, h\gB, q})$ of the Hamiltonians $\rH_{X,h\gB,q}$,
the Neumann realizations  of the differential expressions
  \begin{equation}\label{5.30}
\tau_{X, h\gB, q} = -\frac{\rD^2}{\rD x^2} + q(x) +
h\sum_{k=1}^\infty \gB_k(.,\delta'_k)\delta'_k,\qquad h>0.
 \end{equation}
     \begin{proposition}\label{prop6.3}
Let $\gd^*<\infty$ and $q$ satisfy \eqref{5.2A}. Assume also that $\rH_{X,q}^N\ge 0$.
%
%
%
%
Then:
\begin{itemize}
\item[(i)] $\rH_{X, h\gB, q}$ is lower semi-bounded and  the equality
%
%
     \begin{equation}\label{5.31}
\sigma_{\ess}(\rH_{X, h\gB, q}) = \sigma_{\ess}(\rH_{X,q}^N) 
     \end{equation}
holds for all $h>0$ whenever condition \eqref{4.40} is satisfied.
\item[(ii)] If, in addition,
$\gB =- \gB^-$  and  $\rH_{X, h\gB, q}$ is lower semibounded for some $h=h_0>0$, then   \eqref{4.40}
is also necessary for the validity of \eqref{5.31} for all  $h>0$. 
 \end{itemize}
    \end{proposition}
   \begin{proof}
(i) is immediate from Theorem \ref{thContSpec}.

(ii) Since $\gB_k < 0$, $k\in \N,$ we have $\gb_{X} =  -\gb^-_{X}\le 0$.
Further, since  $q$ satisfies \eqref{5.2A} and  the form $\gt_{X,h_0\gB,q}=\gt_{X,q} - h_0\gb^-_{X}$ is semibounded from below, it follows from Proposition \ref{cor:3.3}(iii) that $\beta = -\beta^-$ satisfies condition  \eqref{eq:b_klmn}. Therefore, by Proposition \ref{cor:3.3}(i), the form $\gt_{X,h\gB,q}$ is lower semibounded and closed for all $h>0$.

According to \eqref{5.31}, the negative spectrum of the form $\gt_{X,q} - h\gb^-_{X}$
is discrete for all $h>0$.
Therefore, by Proposition \ref{prop5.3}(ii) and Remark \ref{rem5.1}(i),
the form $\gb_{X}= -\gb_X^-$
is compact in $\gH_{X,q}$, i.e., it is $\gt_{X,q}$-compact.
Applying Lemma \ref{lem5.1}, we complete the proof.
       \end{proof}
     \begin{corollary}\label{cor6.3}
Let $\gd^*<\infty$ and let $\mathcal{D}$ be given by \eqref{eq:d1} and \eqref{eq:d2}.
Assume  also that $\rH_{X,q}^N\ge 0$ and $q$ satisfies condition  \eqref{5.20}.
Then the equality    
     \begin{equation}\label{5.31A}
 \sigma_{\ess}(\rH_{X,h\gB,q}) = \sigma_{\ess}(\rH_X^N)=
\{0\}\bigcup_{\lambda\in\mathcal{D}}\Big\{\frac{\pi^2 n^2}{\lambda^2}\Big\}_{n=1}^\infty
    \end{equation}
holds for all $h>0$ if $\gB$ satisfies condition  \eqref{4.40}. 
If, in addition, $\gB = -\gB^-$, 
then the condition  \eqref{4.40}
is also necessary for the validity of equalities \eqref{5.31A} for all  $h>0$.
     \end{corollary}
\begin{proof}
The proof follows by combining Proposition \ref{th5.2} with Proposition
\ref{prop6.3}.
  \end{proof}
Next we complete Corollary \ref{cor5.1} by showing that in the case  $d_*>0$ condition
\eqref{5.26} is in a sense  necessary for the $h$-stability of an essential spectrum of the perturbed Hamiltonian.
   \begin{corollary}
Let $0<\gd_*\le \gd^*<\infty$. Let also $\rH_{X,q}^N\ge 0$ and $q$ satisfy \eqref{5.2A}. 
 \begin{itemize}
\item[(i)] Condition \eqref{5.26}
is sufficient, and in the case  $\gB = -\gB^-$       
it is also necessary, for the validity of equalities  \eqref{5.31} for all $h>0$.
\item[(ii)] If, in addition, $q$ satisfies \eqref{5.20},
then condition \eqref{5.26} 
is sufficient, and in the case $\gB = -\gB^-$ 
it is also necessary, for the validity of equalities  \eqref{5.31A} for all $h>0$.
\end{itemize}
   \end{corollary}
\begin{proof}
Since  $d_*>0$ and $\gB_k < 0$, $k\in \N$,
condition \eqref{5.26} is equivalent to \eqref{4.40}.  
Therefore the statements $(i)$ and $(ii)$ are immediate from Proposition
\ref{prop6.3} and Corollary \ref{cor6.3}, respectively.
     \end{proof}

 \appendix

 \section{Quadratic forms}\label{Sec_prelim}

Let $\gH$ be a Hilbert space and let $\gt$ be a densely defined
quadratic form in $\gH$. Let also $\gt$ be lower semibounded,
$\gt[u]\ge -c\|u\|_\gH^2$, $u\in\dom(\gt)$, for some $c\in\R$.
Denote by $\gH'_\gt$ the domain $\dom(\gt)$ equipped  with the
norm
\begin{equation}\label{1.1}
\|u\|_\gt:=\gt[u]+(1+c)\|u\|^2_\gH,\qquad u\in\dom(\gt).
     \end{equation}
The form $\gt$ is called \emph{closable} if the norms $\|\cdot\|$
and $\|\cdot\|_\gt$ are topologically consistent, i.e., for
every Cauchy sequence $\{u_n\}_{n=1}^\infty\subset \dom(\gt)$ with
respect to $\|\cdot\|_\gt$, $\|u_n\|_\gH\to 0$ implies
$\|u_n\|_\gt\to 0$. In this case the completion $\gH_\gt$ of
$\gH'_\gt$ can be realized as a subset of  $\gH.$

The form $\gt$ is called \emph{closed} if $\gH'_\gt = \gH_\gt$.
%
%

Let $A$ be a self-adjoint lower semibounded operator in $\gH$,
$A=A^*\ge -c I_\gH$. Denote by $\gt'_A$ a (densely defined)
quadratic form given by $\gt'_A[f]=(A f,f)$,
$\dom(\gt'_A)=\dom(A)$. It is known (see
\cite{Kato66}) that this form is closable and lower
semibounded, $\gt'_A\ge -c$. Its closure $\gt_A$ satisfies
$\gt_A\ge -c$ and $\gH_\gt:=\dom(\gt_A)
=\dom\bigl((A+c)^{1/2}\bigr)$. Moreover, by the second
representation theorem \cite[Theorem 6.2.23]{Kato66}, $\gt_A$
admits the representation
 \begin{equation}\label{repr_1}
\gt_A[u]= \|(A+c)^{1/2}u\|^2_\gH -c\|u\|^2_\gH, \qquad
u\in\dom(\gt_A). 
\end{equation}
%
%

We set  $\gH_A:= \gH_{\gt_A}$.  Further, by the first
representation theorem \cite[Theorem 6.2.1]{Kato66}, to any closed
lower semibounded quadratic form $\gt\ge -c$ in $\gH$ there corresponds
a unique self-adjoint operator $A=A^*\ge -c$ in $\gH$ such that
$\gt$ is the closure of $\gt'_A$. It is uniquely determined by the
conditions $\dom(A)\subset\dom(\gt)$ and
    \begin{equation}\label{1.2}
(Au,v)=\gt[u,v], \qquad   u\in\dom(A),\  v\in \dom(\gt),
    \end{equation}
where $\gt[\cdot,\cdot]$ is a sesquilinear form associated with $\gt$ via the polarization identity.

The following theorem is well known (see e.g. \cite{Kato66}).

 \begin{theorem} [Rellieh]\label{th_Rel}
Let $A=A^*$ be a lower semibounded operator in $\gH$ and let
$\gt_A$ be the corresponding form. The spectrum $\sigma(A)$ of the
operator $A$ is discrete if and only if the embedding $i_A:\gH_A
\hookrightarrow \gH$ is compact.
  \end{theorem}

\begin{definition}\label{def2.2}
Let the operator $A$ be self-adjoint and positive in $\gH$,
$A=A^*> 0$, and let $\gt_A$ be the corresponding form. The form
$\gt$ is called \emph{relatively form bounded} with respect to
$\gt_A$ ($\gt_A$-bounded) if $\dom(\gt_A)\subseteq \dom(\gt)$ and
there are positive constants $a,b>0$ such that
\begin{equation}\label{II.4}
|\gt[f]|\leq a\gt_A[f]+b\|f\|^2_\gH,\qquad f\in\dom(\gt_A).
\end{equation}
The infimum of all possible $a$ is called \emph{the form bound} of $\gt$ with respect to $\gt_A$. If $a$ can be chosen arbitrary small, then $\gt$ is called \emph{infinitesimally form bounded} with respect to $\gt_A$.
\end{definition}

For the proof of the following theorem see, e.g., \cite{RedSim75}.

\begin{theorem}[KLMN]\label{th_KLMN}
Let $\gt_A$ be the form corresponding to the operator $A=A^*> 0$ in $\gH$. If the form $\gt$ is $\gt_A$-bounded with relative bound $a<1$, then the form
\begin{equation}\label{II.5}
{\gt}_1 := {\gt}_A+{\gt}, \qquad   \dom{\gt}_1:=\dom({\gt}_A),
\end{equation}
is closed and lower semibounded in $\gH$ and hence gives rise to a self-adjoint semibounded operator. 
Moreover, the norms  $\|\cdot\|_A$  and $\|\cdot\|_{\gt_1}$ are
equivalent.
   \end{theorem}

Recall that a quadratic form $\gt$ in $\gH$ is called {\em compact} if it is bounded, $\gt= \gt_C,$
and the (bounded) operator $C$ is compact in $\gH$.

We also need the following result of Birman \cite[Theorem 1.2]{Bir61} (see also
\cite[Theorem 1.19]{Gla65})
  \begin{theorem}[Birman]\label{th2.2}
Let $A=A^*$ be a self-adjoint lower semibounded operator in $\gH$  and let $\gt_A$
 be the corresponding form.
If the quadratic form $\gt$ in ${\gH}$  is compact in $\gH_A$ (or
simply, $\gt_A$-compact), then the form $\gt_1$ defined by
\eqref{II.5}
%
%
%
%
is closed, lower semibounded in $\gH$, and the operator $B=B^*$
associated with the form ${\gt}_1$  satisfies
$\sigma_{\ess}(B)=\sigma_{\ess}(A)$.
    \end{theorem}
\begin{remark}
Note that the form $\gt$ is  infinitesimally $\gt_A$-bounded if it
is $\gt_A$-compact.
\end{remark}

A weaker form of the following lemma is known (cf. \cite[Theorem
1.17]{Gla65}).
       \begin{lemma}\label{lem3.2}
Let $\gt$ and $\gt_1$ be (not necessarily closable)
lower semibounded forms in $\gH$ 
 and assume that
     \begin{equation}\label{3.1}
\big|\gt[u]\big|\le\gt_1[u], \qquad u\in\dom(\gt_1)\subseteq
\dom(\gt).
   \end{equation}
Assume also that $A=A^*\ge -cI$ and $\gt_1$ is $\gt_A$-compact.  Then $\gt$ is $\gt_A$-compact too.
     \end{lemma}

   We also need the following useful fact.
   \begin{lemma}\label{lem3.1}
Let $A=A^*\ge -cI$ and let $\gt$ be a nonnegative (not necessarily
closable) quadratic form in $\gH$.  Assume that $\gH_A \subset
\dom(\gt)$ and $\gt$ is closable  in $\gH_A$. Then the form $\gt$
is compact in $\gH_A$ if and only if the embedding $i:
\gH_A\to\dom(\gt)$ is compact.
  \end{lemma}
Recall that the negative spectrum of a self-adjoint operator  is
called \emph{discrete} if it has at  most two accumulation points
$0$ and $\infty$.
     \begin{proposition}\label{prop5.3}
Let $A=A^*\ge 0$ and $\gt_A$ be the corresponding
form in $\gH$. Assume that $h\in\R_+$ and $\gt_1$ is a nonnegative
(not necessarily closable) infinitesimaly $\gt_A$-bounded form in $\gH$.  Then:
\begin{itemize}
\item[(i)]  the form $\gt(h):=\gt_A-h\gt_1$ is lower semibounded and closed in $\gH$ and
$\dom\bigl(\gt(h)\bigr) = \dom(\gt_A) = \gH_A$.
\item[(ii)]  If, in addition, $\gt_1$ is closed in $\gH_A$, then the negative spectrum of the form
$\gt(h)$ is discrete for every $h>0$ if and only if the form
$\gt_1$ is compact in $\gH_A$.
\end{itemize}
      \end{proposition}

    \begin{remark}\label{rem5.1}
$(i)$  \  Statement $(ii)$ of Proposition \ref{prop5.3} remains valid if we replace
infinitesimally $\gt_A$-boundedness of $\gt_1$ by the assumption that $\dom(\gt_1) \supset \dom(\gt_A)$
and the form $\gt(h)=\gt_A-h\gt_1$ is lower semibounded and closed in $\gH$ for every $h>0$.

$(ii)$
Proposition \ref{prop5.3} was obtained by M.S. Birman \cite[Theorem 1.3]{Bir61} under the assumption that the form $\gt_1$ is closable in $\gH$.
For a proof of Proposition \ref{prop5.3} we refer to \cite{AKM_10}.
         \end{remark}


\end{document}